\pdfoutput=1
\documentclass[11pt,twoside,fleqn]{article}
\usepackage{color,colortbl,graphicx,setspace,multirow,enumitem}

\usepackage[utf8]{inputenc}

\usepackage{geometry}
\usepackage[hang,flushmargin]{footmisc}
\usepackage{natbib} % citations
\usepackage{bibunits}
\defaultbibliography{references}
\defaultbibliographystyle{custom}

\usepackage{amsfonts}
\usepackage{amsmath}
\usepackage{amsthm}
\usepackage{amssymb}
\usepackage{algorithm}
\floatstyle{plaintop}
\restylefloat{algorithm}
\usepackage{placeins}

\newtheorem{proposition}{Proposition}

\usepackage{booktabs,threeparttable} % tables

\usepackage[dvipsnames]{xcolor}
\definecolor{nblue}{HTML}{000660}
\usepackage[colorlinks=true,urlcolor=nblue,linkcolor=nblue,citecolor=nblue,hypertexnames=false]{hyperref}
\newcommand{\inlinehead}[1]{{\textbf{#1}}. }

\newcommand{\diag}{\text{diag}}
\newcommand{\bdiag}{\text{bdiag}}

\usepackage[title,titletoc]{appendix}

\usepackage{titlesec} % Allows customization of titles
\titleformat{\section}[block]{\bfseries\sffamily\large}{\thesection. }{0em}{\MakeUppercase} % Change the look of the section titles
\titleformat{\subsection}[block]{\bfseries\sffamily\large}{\thesubsection. }{0em}{} % Change the look of the section titles
\titleformat{\subsubsection}[block]{\large\sffamily}{}{0em}{\itshape} % Change the look of the section titles

\usepackage{bm}
\usepackage{caption}
\usepackage{subcaption}
\newcommand\cond[1][]{\:#1\vert\:}

\makeatletter\let\p@subfigure\thefigure\makeatother
\newcommand\fnote[1]{\captionsetup{font={small,singlespacing}}\caption*{#1}}

\usepackage{fancyhdr} % Headers and footers
\def\titletext{Conditional projection methods for large-scale Bayesian VARs}

\title{\sffamily\Huge{\textbf{\titletext}}}
\author{}
\date{}

\begin{document}

\maketitle
\vspace*{-4.5em} 
\normalsize
\begin{center}
\begin{minipage}{.45\textwidth}
 \large\centering Niko \MakeUppercase{Hauzenberger}\\[0.25em]
 \small University of Strathclyde
\end{minipage}
\begin{minipage}{.45\textwidth}
 \large\centering Michael \MakeUppercase{Pfarrhofer}\\[0.25em]
 \small WU Vienna
\end{minipage}
\end{center}

\vspace*{1em}
\doublespacing
\begin{center}
\begin{minipage}{0.85\textwidth}
\noindent\small We develop methods for conditional forecasting and structural scenario analysis with high-dimensional Bayesian vector autoregressions (VARs). Our framework features a factor structure on the reduced-form errors, which enables fast and order-invariant equation-by-equation estimation; suitably identified factors admit a structural interpretation. We define scenarios through distributional restrictions on observables, structural shocks and idiosyncratic components. The computational cost of our proposed algorithm is cubic only in the number of restrictions, while the dimension of the forecasted system enters linearly. In our application with $37$ variables for the US and ten structural shocks, identified by sign restrictions and internal instruments, we compute counterfactual predictions for oil price scenarios in the context of the 2026 closure of the Strait of Hormuz. Left to the model, the oil price path is read as an oil price shock accompanied by a demand boom; attributed to the structural oil price shock alone, a supply disruption, it implies higher inflation and slower growth.

\end{minipage}
\end{center}

% acknowledgements
\singlespacing\vfill\noindent{\footnotesize\textit{Contact}: \href{mailto:michael.pfarrhofer@wu.ac.at}{michael.pfarrhofer@wu.ac.at}, Department of Economics, WU Vienna University of Economics and Business. Codes available at \url{https://github.com/mpfarrho/ssa-var-fm}. We used Claude Code and Codex for drafting, editing and brainstorming. We reviewed and edited the content as needed and take full responsibility for the final content and accuracy of this paper.}

\thispagestyle{empty}\renewcommand{\footnotelayout}{\setstretch{1.5}}\newpage
\doublespacing\normalsize\renewcommand{\thepage}{\arabic{page}}

\begin{bibunit}
\section{Introduction}\label{sec:introduction}
Conditional projection methods have been a key component of macroeconometrics since at least \citet{doan1984forecasting} and underpin the modern practice of counterfactual policy and scenario analysis \citep[see, e.g.,][]{crump2021large,adrian2026risks}. Fully specified structural models commit the analysis to a single model by design and often restrict the size of the usable information set. An alternative is to use structural vector autoregressions (VARs), in the spirit of the structural scenario analysis (SSA) of \citet{antolin2021structural}, in which scenarios are implemented as distributional restrictions on observables and structural shocks. Restricting subsets of the latter to their unconditional distribution defines the structural sources of the scenario, attributing it to identified economic forces. Our paper contributes to this literature by developing a high-dimensional Bayesian VAR for structural analysis to obtain conditional predictions with fast algorithms.

Rich information sets have proven successful in both forecasting and structural settings when combined with appropriate regularization \citep[see][among others]{banbura2010large}. Our framework is a Bayesian VAR (BVAR) with a global--local shrinkage prior on the coefficients and a factor structure on the reduced-form errors. This setup serves two purposes. First, the prior and factor structure enable regularization alongside quick and order-invariant estimation of VARs with a large cross section of $n$ endogenous variables \citep[see][]{kastner2020sparse} and long lag orders. The latter can reduce bias in dynamic multiplier estimates at longer horizons \citep[see the local projection versus VAR discussion, e.g., in][]{Baumeister2025comment}, which matters here because our conditional projections are built from the same moving-average matrices as impulse response functions (IRFs). Second, following \citet{korobilis2022new}, a small number $q\,(\ll n)$ of ``primitive'' shocks --- the factors --- drive the co-movement of the reduced-form errors, while cross-sectionally independent components capture idiosyncratic shocks; when the factors are orthogonal and identified, they admit an interpretation as structural economic shocks.\footnote{\citet{arias2026large} and \citet{chan2025largesign} discuss high-dimensional structural VAR (SVAR) methods for unrestricted covariance matrices, without the reduced-rank plus idiosyncratic structure. When $q = n$, our framework is close to a conventional SVAR; irrespective of the underlying model, the conditional projection algorithm we develop is directly applicable. The factor assumption about the reduced-form errors is used, e.g., in \citet{banbura2026drives,chan2022large,gambetti2023agreed,pruser2024large}.}

Earlier conditional projection algorithms \citep[in the spirit of][]{waggoner1999conditional} can be computationally intensive and thus are often limited to smaller dynamic systems; scalability to richer information sets is possible through filtering-based methods \citep[see][]{banbura2015conditional} or the closed-form solutions of \citet{antolin2021structural} and \citet{chan2023conditional} for joint predictive distributions, on which we build.

We make three contributions: (i) a framework with distributional restrictions on observables, structural shocks, and idiosyncratic components, so that a scenario can be attributed to chosen structural shocks rather than absorbed by idiosyncratic components or other shocks, together with a characterization of which conditional forecasts depend on the identification under each attribution; (ii) an exact sampler whose cost scales with the number of restrictions rather than with the size of the forecasted system; and (iii) an empirical demonstration that the structural attribution of a scenario shapes the counterfactual predictions, together with divergence measures that quantify the cost of each attribution. We illustrate these aspects using a quarterly dataset of $37$ US variables, in which sign restrictions and internal instruments identify ten structural shocks. The scenarios are taken from \citet{kilian2026impact} and impose a path for the nominal price of oil in the context of the 2026 Iran War and the closure of the Strait of Hormuz. Successively adding restrictions on the idiosyncratic components and non-driving structural shocks, we trace how the shocks required to deliver the same oil price path change and how this choice shapes the counterfactual predictions for the unrestricted observables. The model reads the same scenario as an oil price shock accompanied by a demand boom when its structural origins are unrestricted, and implies higher inflation and slower growth, with median growth remaining positive, when the scenario is attributed to the structural oil price shock alone (a supply disruption in our model). The latter requires an oil price shock of about six standard deviations on impact, which markedly increases the divergence of the required innovations from their unconditional distribution. Finally, we explore counterfactuals with monetary policy tightenings of different sizes in response to the oil price induced inflationary pressure.

\section{Econometric Framework}\label{sec:econometrics}
Consider a reduced-form VAR for $n$ variables $\bm{y}_t = (y_{1t},\hdots,y_{nt})'$ with $P$ lags in $\bm{x}_t = (\bm{y}_{t-1}',\hdots,\bm{y}_{t-P}')'$:
\begin{equation}
    \bm{y}_t = \bm{a} + \sum_{p=1}^P \bm{A}_p \bm{y}_{t-p} + \bm{\eta}_t = \bm{a} + \bm{A}\bm{x}_t + \bm{\eta}_t, \quad \bm{\eta}_t \sim \mathcal{N}(\bm{0}_n,\bm{\Sigma});\label{eq:model-main}
\end{equation}
$\bm{a}$ is an intercept, $\bm{A} = [\bm{A}_1,\hdots,\bm{A}_P]$ comprises the dynamic coefficient matrices $\bm{A}_p$, and $\bm{\eta}_t$ is a zero-mean reduced-form error with covariance matrix $\bm{\Sigma}$. 

Since an infinite number of structural models is consistent with a reduced-form representation as in (\ref{eq:model-main}), structural economic interpretations require further restrictions, which are commonly imposed via decompositions of the covariance matrix $\bm{\Sigma}$. We use a factor model for the reduced-form error, because this specification is computationally attractive and encompasses a range of possible avenues for structural identification:
\begin{equation}
    \bm{\eta}_t = \bm{L}\bm{\varepsilon}_t + \bm{S}^{1/2}\bm{u}_t, \quad (\bm{\varepsilon}_t', \bm{u}_t')'\overset{\text{iid}}{\sim}\mathcal{N}(\bm{0}_{q+n},\bm{I}_{q+n}),\label{eq:model-error}
\end{equation}
where the $n \times q$ matrix $\bm{L}$ collects the loadings associated with $q$ factors $\bm{\varepsilon}_t$, and $\bm{S} = \diag(s_1^2,\hdots,s_n^2) = \bm{S}^{1/2}\bm{S}^{1/2}$ is a diagonal matrix collecting the variances of the idiosyncratic component. This yields the decomposition $\bm{\Sigma} = \bm{L}\bm{L}' + \bm{S}$ \citep[as in][]{anderson1956statistical}. Any co-movement in the reduced-form errors is due to the common factors $\bm{\varepsilon}_t$, while the idiosyncratic component is cross-sectionally uncorrelated. For $n = q$, $\bm{S} = \bm{0}_{n\times n}$, and $\bm{L}$ invertible, premultiplying by $\bm{A}_0 = \bm{L}^{-1}$ results in the standard SVAR $\bm{A}_0\bm{y}_t = \bm{A}_0(\bm{a} + \bm{A}\bm{x}_t) + \bm{\varepsilon}_t$ and $\bm{\varepsilon}_t$ admits a structural interpretation.

In canonical SVARs, there are as many shocks as observables. Factor models are a dimension-reduction technique with $q \ll n$, based on the idea that only a few economic driving forces are responsible for the majority of the observed business cycle fluctuations. In this case, assuming $\bm{L}$ has full column rank, obtaining an approximate structural form of the model involves the pseudoinverse $\bm{A}_0 = (\bm{L}'\bm{L})^{-1}\bm{L}'$, to yield $\bm{A}_0\bm{y}_t = \bm{A}_0(\bm{a} + \bm{A}\bm{x}_t) + \bm{\varepsilon}_t + (\bm{L}'\bm{L})^{-1}\bm{L}'\bm{S}^{1/2}\bm{u}_t$. The last term is a weighted average of $n$ independent idiosyncratic components and becomes negligible as the number of variables loading on the factors grows. Along these lines, \citet{korobilis2022new} argues that it vanishes for each $t$ as $n\rightarrow\infty$, so the factors can be recovered from the reduced-form errors; once identified, they admit an interpretation as structural shocks.

\inlinehead{Identifying assumptions} The decomposition $\bm{\Sigma} = \bm{L}\bm{L}' + \bm{S}$ raises two identification problems. The first concerns the split of $\bm{\Sigma}$ into $\bm{L}\bm{L}'$ and $\bm{S}$. For $q$ small relative to $n$ the split is generically identified \citep[see][]{anderson1956statistical}, which is the case in our application; otherwise several pairs $(\bm{L}\bm{L}', \bm{S})$ are consistent with the same $\bm{\Sigma}$, and the prior selects among them. The second problem is the rotation. Given $(\bm{L}\bm{L}', \bm{S})$, $\bm{L}$ itself is not identified, since for any $q\times q$ orthogonal matrix $\bm{Q}$, the matrix $\tilde{\bm{L}} = \bm{L}\bm{Q}$ yields the same product $\bm{L}\bm{L}' = \tilde{\bm{L}}\tilde{\bm{L}}'$. The likelihood depends on $\bm{L}$ only through $\bm{L}\bm{L}'$, so the data cannot distinguish between the elements of the identified set $\{\bm{L}\bm{Q}\}$, and a structural interpretation requires restrictions on $\bm{L}$ that select among them. Zero restrictions can achieve point identification at the cost of potentially arbitrary exclusion restrictions; by contrast, sign restrictions lead to set-identification. There is a long-standing tradition to set-identify structural models via sensible sign and zero restrictions \citep[see][]{baumeister2015sign,arias2018inference} which we adopt.

Following \citet{korobilis2022new}, each loading receives a pre-specified restriction $\mathcal{R}_{ij} \in \{+, -, c_{ij}, \text{NA}\}$, collected in the $n\times q$ matrix $\bm{\mathcal{R}}$, which the prior on $\bm{L}$ implements: sign restrictions ($+$, $-$) via truncated normal priors on the respective loadings, constants $c_{ij} \in \mathbb{R}$ via point masses (with $c_{ij} = 0$ yielding zero restrictions), and unrestricted loadings ($\text{NA}$) via normal priors. In terms of the identified set, this prior assigns zero mass to rotations that violate the restrictions and weights the remaining ones according to the normal priors on the free loadings. Since the likelihood is flat across the set, the posterior of $\bm{L}$ given $\bm{L}\bm{L}'$ and $\bm{S}$ is determined by the prior. The data update the reduced-form covariance and its split into common and idiosyncratic parts, but not the position within the identified set.\footnote{Our sampling algorithm does not rotate the loadings; the matrix $\bm{Q}$ merely indexes the set. We revisit these aspects in Section \ref{sec:ssa} when we ask which outputs of the scenario analysis vary across that set.}

\inlinehead{Prior choices and Bayesian estimation} Before turning to conditional predictions, we briefly summarize prior choices and the Markov chain Monte Carlo (MCMC) algorithm. We mostly follow earlier papers here, and details about the posteriors and the standard Gibbs sampling algorithm are provided in Online Appendix \ref{app:oa-technical} (henceforth OA). The number of parameters in the conditional mean is $k + n$ where $k = n^2P$, which quickly exceeds $T$ even for moderate $n$. This typically results in overfitting if estimated without any regularization, and can be computationally slow. We shrink via a global--local prior (horseshoe) and rely on equation-by-equation estimation \citep[see, e.g.,][]{korobilis2022new}. The factor structure on the reduced-form errors in (\ref{eq:model-error}) enables fast and order-invariant equation-by-equation estimation, since (\ref{eq:model-main}) reduces to $n$ independent equations conditional on the factors. The idiosyncratic variances in $\bm{S}$ receive inverse Gamma priors, and the normal priors on the free loadings are weakly informative.

This baseline approach can be combined with additional identifying restrictions and/or model features.\footnote{Related frameworks exploit heteroskedasticity \citep{chan2022large} or non-Gaussian features of the structural shocks \citep{pruser2024large} for identification, or restrict specific shocks in specific periods \citep{berend2025large}; external and internal instruments enter as observed or partially latent factors \citep{korobilis2025exploring,pfarrhofer2025high}, which we do in our application.} We abstract from most of these options, noting that such extensions require only minor adjustments for use with our proposed SSA algorithm.

\subsection{Structural Scenario Analysis}\label{sec:ssa}
The joint vector of forecasts is $\bm{y}_{T+(1:H)} = (\bm{y}_{T+1}',\hdots,\bm{y}_{T+H}')'$, where $H$ refers to the maximum forecast horizon. In the same vein, we stack $\bm{\varepsilon}_{T+(1:H)} = (\bm{\varepsilon}_{T+1}',\hdots,\bm{\varepsilon}_{T+H}')'$ and idiosyncratic errors $\bm{u}_{T+(1:H)} = (\bm{u}_{T+1}',\hdots,\bm{u}_{T+H}')'$ in a joint vector $\bm{\xi}_{T+(1:H)} = (\bm{\varepsilon}_{T+(1:H)}',\bm{u}_{T+(1:H)}')' \sim\mathcal{N}(\bm{0},\bm{I}_{d})$ with $d = (q+n)H$, based on (\ref{eq:model-error}). Precise definitions of the vectors and matrices below are provided in Appendix \ref{app:technical}. Using $\bm{\mathcal{L}}_{(H)} = (\bm{I}_H \otimes \bm{L})$, $\bm{\mathcal{S}}^{1/2}_{(H)} = (\bm{I}_H \otimes \bm{S}^{1/2})$ and $\bm{\mathcal{V}} = [\bm{\mathcal{L}}_{(H)}, \bm{\mathcal{S}}_{(H)}^{1/2}]$ of size $nH \times d$, this allows us to write: $\bm{H}\bm{y}_{T+(1:H)} = \bm{h} + \bm{\mathcal{L}}_{(H)}\bm{\varepsilon}_{T+(1:H)} + \bm{\mathcal{S}}^{1/2}_{(H)}\bm{u}_{T+(1:H)} = \bm{h} + \bm{\mathcal{V}}\bm{\xi}_{T+(1:H)}$. Thus:
\begin{equation}
    \bm{y}_{T+(1:H)} = \bm{H}^{-1}\bm{h} + \bm{H}^{-1}\bm{\mathcal{V}}\bm{\xi}_{T+(1:H)},\label{eq:jointfcdist}
\end{equation}
where below we also use $\bm{m} = \bm{H}^{-1}\bm{h}$ and $\bm{M}' = \bm{H}^{-1}\bm{\mathcal{V}}$, so $\bm{y}_{T+(1:H)}\sim\mathcal{N}(\bm{m}, \bm{M}'\bm{M})$. The vector $\bm{h}$ collects initial conditions and any deterministic terms, and $\bm{H}$ encodes the dynamic VAR coefficients in a banded structure.

The inverse $\bm{H}^{-1}$ collects the coefficient matrices of the vector moving average (VMA) representation, $\bm{\Phi}_j = \bm{J}\bm{\mathcal{A}}^{j}\bm{J}'$ with $\bm{\Phi}_0 = \bm{I}_{n}$, defined using $\bm{J} = [\bm{I}_n,\bm{0}_{n\times n(P-1)}]$ and the companion form $\bm{\mathcal{A}} = [\bm{A}; \bm{I}_{n(P-1)}, \bm{0}_{n(P-1)\times n}]$. $\bm{H}^{-1}$ features copies of $\bm{\Phi}_0$ on its main diagonal, and the $h$th block subdiagonal for $h = 1,2,\hdots,H-1$, consists of copies of $\bm{\Phi}_h$. Intuitively, consider the IRF of the endogenous variables to a structural shock. The IRF to a structural shock in $\bm{\varepsilon}_{T+1}$ is defined as $\partial \bm{y}_{T+h+1} / \partial \bm{\varepsilon}_{T+1}' =  \bm{\Phi}_h\bm{L}$ for $h = 0, 1, \hdots, H - 1$, with the timing convention that the impact response $h = 0$ refers to period $T+1$. The subdiagonal structure of $\bm{H}^{-1}$ propagates forward any vector it multiplies according to the dynamics of the model, e.g., structural shock impacts at each horizon, via $\bm{H}^{-1}\bm{\mathcal{L}}_{(H)}$ in (\ref{eq:jointfcdist}). The example showcases that IRFs and the associated VMA matrices are the fundamental building blocks of any conditional projection and scenario analysis.

Stochastic restrictions as in \citet{andersson2010density} are imposed on \textit{observed variables}, where there are $r_y$ restrictions, and $\bm{R}^{(y)}$ is $r_y \times nH$:
\begin{equation}
    \bm{R}^{(y)}\bm{y}_{T+(1:H)} \sim \mathcal{N}(\bm{r}^{(y)}, \bm{\Omega}^{(y)}),\label{eq:rest_y}
\end{equation}
on the \textit{structural shocks}: $\bm{R}^{(\varepsilon)}\bm{\varepsilon}_{T+(1:H)} \sim \mathcal{N}(\bm{r}^{(\varepsilon)},\bm{\Omega}^{(\varepsilon)})$ with $r_{\varepsilon}$ restrictions and $\bm{R}^{(\varepsilon)}$ is $r_{\varepsilon} \times qH$; and, generalizing earlier papers, on the \textit{idiosyncratic component}: $\bm{R}^{(u)}\bm{u}_{T+(1:H)} \sim \mathcal{N}(\bm{r}^{(u)},\bm{\Omega}^{(u)})$ with $r_{u}$ restrictions and $\bm{R}^{(u)}$ is $r_{u} \times nH$. The latter are stacked using $\bm{R}^{(\xi)} = \bdiag(\bm{R}^{(\varepsilon)}, \bm{R}^{(u)})$, $\bm{r}^{(\xi)} = [\bm{r}^{(\varepsilon)}; \bm{r}^{(u)}]$ and $\bm{\Omega}^{(\xi)} = \bdiag(\bm{\Omega}^{(\varepsilon)},\bm{\Omega}^{(u)})$:
    \begin{equation}
        \bm{R}^{(\xi)}\bm{\xi}_{T+(1:H)} \sim \mathcal{N}(\bm{r}^{(\xi)}, \bm{\Omega}^{(\xi)}).\label{eq:rest_xi}
    \end{equation}
Restrictions on the idiosyncratic component are the counterpart of those on the structural shocks. The components $\bm{u}_t$ are orthogonal to the common shocks and to each other; they collect series-specific disturbances, possibly including measurement error, that the model does not interpret structurally. A scenario that is to be read in terms of economic driving forces should not be delivered by them, and restricting $\bm{u}_{T+(1:H)}$ to its unconditional distribution over the restricted horizons closes this margin.

Plugging (\ref{eq:jointfcdist}) into (\ref{eq:rest_y}) yields $\bm{R}^{(y)}\left(\bm{m} + \bm{M}'\bm{\xi}_{T+(1:H)}\right) \sim \mathcal{N}(\bm{r}^{(y)}, \bm{\Omega}^{(y)})$, which we combine with (\ref{eq:rest_xi}). There are a total of $r = r_y + r_\varepsilon + r_u$ restrictions. Defining $\bm{R} = [\bm{R}^{(y)}\bm{M}'; \bm{R}^{(\xi)}]$ of size $r \times d$, which is assumed to have full row rank (no redundancies), $\bm{r} = [\bm{r}^{(y)} - \bm{R}^{(y)}\bm{m}; \bm{r}^{(\xi)}]$ which is $r \times 1$, and $\bm{\Omega} = \bdiag(\bm{\Omega}^{(y)}, \bm{\Omega}^{(\xi)})$ which is $r \times r$, we obtain $\bm{R}\bm{\xi}_{T+(1:H)} \sim \mathcal{N}(\bm{r}, \bm{\Omega})$, which specifies the distribution that $\bm{R}\bm{\xi}_{T+(1:H)}$ is required to have under the conditional forecast distribution.

\inlinehead{Conditioning schemes} The restrictions on observables state the scenario, and the restrictions on innovations determine how it is delivered. We distinguish four conditioning schemes that combine a given set of restrictions on observables with different restrictions on the innovations. Under the \textit{reduced-form} scheme (RF), the restrictions on observables are the only restrictions, and the scenario is delivered by whatever combination of structural shocks and idiosyncratic components is most likely under the model's unconditional distribution. The \textit{structural} scheme (ST) closes the idiosyncratic margin, forcing all idiosyncratic components to their unconditional distribution, so that the scenario must be delivered by the structural shocks. The \textit{driving-shock} scheme (DS) in addition restricts all structural shocks but a designated set of driving shocks to their unconditional distribution, that is, only the driving shocks may deviate.\footnote{Restricting innovations to their unconditional distribution amounts to $\bm{r} = \bm{0}$ and $\bm{\Omega} = \bm{I}$ of conformable dimension, at the horizons spanned by the restrictions on observables. Innovations at later horizons do not enter the restricted combinations and retain their unconditional distribution because $\bm{H}^{-1}$ is block lower-triangular, so observables at a given horizon depend on innovations up to that horizon only.} The fourth combination, the \textit{open-margin} scheme (OM), restricts the non-driving structural shocks but leaves the idiosyncratic component free. It designates a driving shock without requiring that the scenario be delivered by it, since the idiosyncratic components of the restricted observables can absorb the scenario instead, and we therefore report it mostly in OA \ref{app:oa-empirical}.

The schemes differ in what they require from the identification. Since the data determine $(\bm{L}\bm{L}', \bm{S})$ and the prior the position within the identified set, the following proposition, proved in Appendix \ref{app:proof-schemes}, states which conditional forecasts vary across this set.

\begin{proposition}\label{prop:schemes}
Fix the reduced-form parameters, initial conditions, and restriction targets. Let $\bm{y}^{\ast}$ denote the conditional projection.
\begin{enumerate}
    \item[(i)] Under the RF scheme, the distribution of $\bm{y}^{\ast}$ depends on the decomposition only through $\bm{\Sigma}$: it is the same for every $(\bm{L}, \bm{S})$ with $\bm{L}\bm{L}' + \bm{S} = \bm{\Sigma}$, and in particular for every rotation $\bm{L}\bm{Q}$.
    \item[(ii)] Under the ST scheme, the distribution of $\bm{y}^{\ast}$ depends on the decomposition $(\bm{L}\bm{L}', \bm{S})$ but not on the rotation $\bm{Q}$.
    \item[(iii)] Under the DS scheme (and OM), the distribution of $\bm{y}^{\ast}$ is not invariant to the rotation $\bm{Q}$ in general.
\end{enumerate}
\end{proposition}

The conditional forecast under the RF scheme is thus identification-free, the ST scheme requires the split of the reduced-form covariance into its common and idiosyncratic parts but not the rotation, and the DS scheme depends on the rotation only through the columns of the driving shocks, which are determined by the prior over the identified set unless they are pinned down by additional restrictions, as in our application. The attribution of the scenario to shocks, by contrast, depends on the identification under every scheme: under the RF and ST schemes, a rotation $\tilde{\bm{L}} = \bm{L}\bm{Q}$ leaves $\bm{y}^{\ast}$ unchanged but replaces the restricted structural shocks $\bm{\varepsilon}^{\ast}$ by $(\bm{I}_H \otimes \bm{Q}')\bm{\varepsilon}^{\ast}$, so the same forecast is attributed to rotated shocks.

\subsection{Sampling the Restricted Innovations}\label{sec:sampling}
The restricted distribution of the innovations is the revision of their unconditional distribution that satisfies the restrictions while minimizing the Kullback--Leibler (KL) divergence to the unconditional distribution (we discuss this in more detail below in the context of measuring the plausibility of scenarios); \citet{antolin2021structural} prove that in the Gaussian case this is equivalent to entropic tilting as in \citet{robertson2005forecasting}. To derive it, we drop the time subscripts for notational convenience and write $\bm{\xi}=\bm{\xi}_{T+(1:H)}\sim\mathcal{N}(\bm{0},\bm{I}_d)$ by the assumptions in (\ref{eq:model-error}), and $\bm{w}=\bm{R}\bm{\xi}$ for the $r$-dimensional restricted combination. The pair $(\bm{\xi},\bm{w})$ is jointly Gaussian with $\mathrm{Cov}(\bm{\xi},\bm{w})=\bm{R}'$ and $\mathrm{Var}(\bm{w})=\bm{R}\bm{R}'$, which is invertible since $\bm{R}$ has full row rank. The usual formula for conditional Gaussian distributions gives:
\begin{equation*}
    \bm{\xi}\mid\bm{w}\sim\mathcal{N}\big(\bm{R}^{+}\bm{w},\ \bm{I}_d-\bm{P}\big), \quad \bm{R}^{+}=\bm{R}'(\bm{R}\bm{R}')^{-1}, \quad \bm{P}=\bm{R}^{+}\bm{R},
\end{equation*}
where $\bm{R}^{+}$ is the Moore--Penrose inverse of $\bm{R}$ (a right inverse, $\bm{R}\bm{R}^{+}=\bm{I}_r$) and $\bm{P}$ is the orthogonal projection onto the row space of $\bm{R}$; it is symmetric and idempotent. Equivalently, $\bm{\xi}=\bm{R}^{+}\bm{w}+\bm{\xi}_\perp$ where $\bm{\xi}_\perp=(\bm{I}_d-\bm{P})\bm{\xi}\sim\mathcal{N}(\bm{0},\bm{I}_d-\bm{P})$ independent of $\bm{w}$. The revision leaves this conditional unchanged and replaces the marginal of $\bm{w}$, $\mathcal{N}(\bm{0},\bm{R}\bm{R}')$, by $\bm{w}^{\ast}\sim\mathcal{N}(\bm{r},\bm{\Omega})$, which defines the restricted innovations $\bm{\xi}^{\ast}=\bm{R}^{+}\bm{w}^{\ast}+\bm{\xi}_\perp$. Hence $\bm{\xi}^{\ast}\sim\mathcal{N}(\bm{\mu}_\xi, \bm{V}_\xi)$, with moments:
\begin{equation}
    \bm{V}_\xi = (\bm{I}_{d} - \bm{P}) + \bm{R}^{+}\bm{\Omega}\bm{R}^{+\prime}, \quad \bm{\mu}_\xi = \bm{R}^{+}\bm{r}.\label{eq:restrmoments}
\end{equation}
By construction, using $\bm{R}(\bm{I}_d-\bm{P})=\bm{0}$, we have $\bm{R}\bm{\mu}_\xi=\bm{r}$ and $\bm{R}\bm{V}_\xi\bm{R}'=\bm{\Omega}$. This approach draws the full sequence of innovations $\bm{\xi}^{\ast}$ compatible with the restrictions. Inserting the restricted $\bm{\xi}^{\ast}$ into (\ref{eq:jointfcdist}) yields the final conditional projection, a draw from $\bm{y}_{T+(1:H)} \cond \bm{R},\bm{r},\bm{\Omega} \sim \mathcal{N}(\bm{\mu}_y, \bm{V}_y)$ with $\bm{\mu}_y = \bm{m} + \bm{M}'\bm{\mu}_{\xi}$ and $\bm{V}_y = \bm{M}'\bm{V}_{\xi}\bm{M}$ in terms of the observables.\footnote{These moments coincide with \citet{antolin2021structural}. OA \ref{app:oa-technical} discusses interpreting the restrictions as noisy measurements of $\bm{R}\bm{\xi}$ and a Bayes' update (combining rather than replacing the model-based forecast); this coincides for $\bm{\Omega} = \bm{0}$ but differs otherwise.} Drawing directly from $\mathcal{N}(\bm{\mu}_\xi,\bm{V}_\xi)$ requires factorizing the $d\times d$ matrix $\bm{V}_\xi$. Algorithm~\ref{alg:xi_post} samples from the same distribution using only the $r\times r$ matrices $\bm{R}\bm{R}'$ and $\bm{\Omega}$; Proposition \ref{prop:algorithm1} formalizes its properties.

\par\medskip\noindent\refstepcounter{algorithm}\label{alg:xi_post}\inlinehead{Algorithm \thealgorithm: Conditional forecast sampling}
\begin{enumerate}[label=(\arabic*),nosep,leftmargin=*]
    \item Sample $\bm{\xi}_{0} \sim \mathcal{N}(\bm{0}_{d},\bm{I}_{d})$ and $\bm{v}_0 \sim \mathcal{N}(\bm{0}_r,\bm{\Omega})$ independently;
    \item Solve $(\bm{R}\bm{R}')\bm{x} = \bm{r} + \bm{v}_0 - \bm{R}\bm{\xi}_{0}$ for $\bm{x}$;
    \item Return $\bm{\xi}^{\ast} = \bm{\xi}_{0} + \bm{R}'\bm{x} = \bm{\xi}_{0} + \bm{R}'(\bm{R}\bm{R}')^{-1}(\bm{r} + \bm{v}_0 - \bm{R}\bm{\xi}_{0})$.
\end{enumerate}

\begin{proposition}\label{prop:algorithm1}
    The output $\bm{\xi}^{\ast}$ of Algorithm~\ref{alg:xi_post} satisfies $\bm{\xi}^{\ast}\sim\mathcal{N}(\bm{\mu}_\xi,\bm{V}_\xi)$, and $\bm{R}\bm{\xi}^{\ast}=\bm{r}+\bm{v}_0\sim\mathcal{N}(\bm{r},\bm{\Omega})$.
\end{proposition}

\begin{proof}
Rearranging step (3) in terms of $\bm{R}^{+}$ and $\bm{P}$ gives
\begin{equation}
    \bm{\xi}^{\ast}=\bm{R}^{+}(\bm{r}+\bm{v}_0)+(\bm{I}_d-\bm{P})\bm{\xi}_0.\label{eq:split-xi}
\end{equation}
With $\bm{w}^{\ast}=(\bm{r}+\bm{v}_0)\sim\mathcal{N}(\bm{r},\bm{\Omega})$ and $\bm{\xi}_\perp=[(\bm{I}_d-\bm{P})\bm{\xi}_0]\sim\mathcal{N}(\bm{0},\bm{I}_d-\bm{P})$, independent since $\bm{v}_0\perp\bm{\xi}_0$, this is the representation $\bm{\xi}^{\ast}=\bm{R}^{+}\bm{w}^{\ast}+\bm{\xi}_\perp$ above: $\bm{\xi}^{\ast}$ is Gaussian with $\mathbb{E}(\bm{\xi}^{\ast}) = \bm{R}^{+}\bm{r} = \bm{\mu}_\xi$ and $\mathrm{Var}(\bm{\xi}^{\ast}) = (\bm{I}_d-\bm{P})^2 + \bm{R}^{+}\bm{\Omega}\bm{R}^{+\prime} = \bm{V}_\xi$ by idempotence of $\bm{I}_d-\bm{P}$. Premultiplying (\ref{eq:split-xi}) by $\bm{R}$ and using $\bm{R}\bm{R}^{+}=\bm{I}_r$ and $\bm{R}(\bm{I}_d-\bm{P})=\bm{0}$ gives $\bm{R}\bm{\xi}^{\ast}=\bm{w}^{\ast}$.
\end{proof}

The algorithm is related to \citet[][their Algorithm 2]{cong2017fast}, and to \citet{jarocinski2010conditional}, who uses a singular value decomposition (SVD) of the matrix $\bm{R}$, instead of orthogonal projections, to speed up the computations described in \citet{waggoner1999conditional}. The draw $\bm{v}_0$ takes the role of the Monte Carlo draw from the distribution of the soft constraints in \citet{waggoner1999conditional}: the realized target $\bm{w}^{\ast} = \bm{r} + \bm{v}_0$ is imposed exactly; hard constraints of the form $\bm{R}\bm{\xi} = \bm{r}$ are nested trivially for $\bm{v}_0 = \bm{0}_r$.

\inlinehead{Plausibility of a scenario} The KL divergence of the restricted from the unconditional innovation distribution measures how unusual the innovations are that a scenario requires, and it is available from the same $r \times r$ quantities. Because the revision changes only the distribution of the restricted combination $\bm{w}=\bm{R}\bm{\xi}$ and leaves the distribution of $\bm{\xi}$ given $\bm{w}$ unchanged, the divergence between the restricted and unconditional distributions of $\bm{\xi}$ reduces to the divergence between the two distributions of $\bm{w}$, $\mathcal{N}(\bm{r}, \bm{\Omega})$ from $\mathcal{N}(\bm{0}, \bm{R}\bm{R}')$. OA \ref{app:oa-technical} provides the derivation and special cases:
\begin{equation}
    \mathrm{KL} = \tfrac{1}{2}\, \bm{r}'(\bm{R}\bm{R}')^{-1}\bm{r} + \tfrac{1}{2}\big[\mathrm{tr}\big((\bm{R}\bm{R}')^{-1}\bm{\Omega}\big) - r - \ln\det\big((\bm{R}\bm{R}')^{-1}\bm{\Omega}\big)\big].\label{eq:kl}
\end{equation}
Following \citet{antolin2021structural}, who propose the divergence as a measure of the plausibility of a scenario, we use it in Section \ref{sec:application} to compare alternative attributions of a scenario, together with their calibration $\kappa = \tfrac{1}{2}\big(1 + \sqrt{1 - e^{-2\mathrm{KL}/d}}\big)$, which translates the divergence into the probability of success of a biased coin such that the divergence of fair from biased coin flips equals $\mathrm{KL}$. Values close to $0.5$ indicate innovations that are hardly distinguishable from their unconditional distribution, values close to one a distortion. Since the conditioning schemes of Section \ref{sec:ssa} add restrictions successively, the divergence is non-decreasing from the RF to the ST to the DS scheme by construction; its increments show how much each additional restriction adds.

\section{Computational Comparison}\label{sec:comptime}
Algorithm~\ref{alg:xi_post} factorizes the $r \times r$ matrix $\bm{R}\bm{R}'$, so the cost of each draw is cubic only in the number of restrictions, while the size of the forecasted system enters linearly, $\mathcal{O}(r^2nH + r^3)$. We benchmark its per-draw cost against four competing samplers of the same conditional distribution. So that all methods are applicable, the comparison is conducted in the canonical SVAR special case of our framework ($q = n$ structural shocks and no idiosyncratic components), for which the competing algorithms were originally designed. The experimental grid varies the number of variables $n \in \{10, 20, 50\}$, the lag order $P \in \{4, 12, 24\}$, and the forecast horizon $H \in \{8, 12, 24\}$. Soft restrictions are imposed on the observables by conditioning a fixed share of the variables at every horizon, so that the restriction ``density'' $r/(nH)$ takes values of $10$, $30$, $50$, and $80$ percent. Costs are accounted for on a per-iteration basis, mimicking use with an MCMC sampler; the total cost per draw is the sum of a construction step, which must be repeated for every draw of the VAR parameters, and the conditioning step. All methods run in \texttt{R}, in a single session on identical hardware (MacBook Air, M1).

The competitors are direct implementations of the formulas in \citet{antolin2021structural} in the space of structural shocks and of observables (``Joint (shocks)'' and ``Joint (obs.)''), which factorize a dense $nH \times nH$ covariance matrix for every draw, an $\mathcal{O}((nH)^3)$ operation; the algorithm of \citet{jarocinski2010conditional} (``SVD''), whose per-draw cost is of order $\mathcal{O}(r(nH)^2)$; and the sampler of \citet{chan2023conditional} (``Precision''), which requires one sparse factorization of the free-block precision of dimension $nH - r$ per draw. Implementation details, and the steps we took to make the comparison as fair as possible, are given in OA \ref{app:oa-technical}. Figure \ref{fig:comp-time} reports total times per draw relative to Algorithm~\ref{alg:xi_post}, whose row shows absolute numbers in seconds per 1,000 draws.

\begin{figure}[ht]
  \centering
  \includegraphics[width=\textwidth]{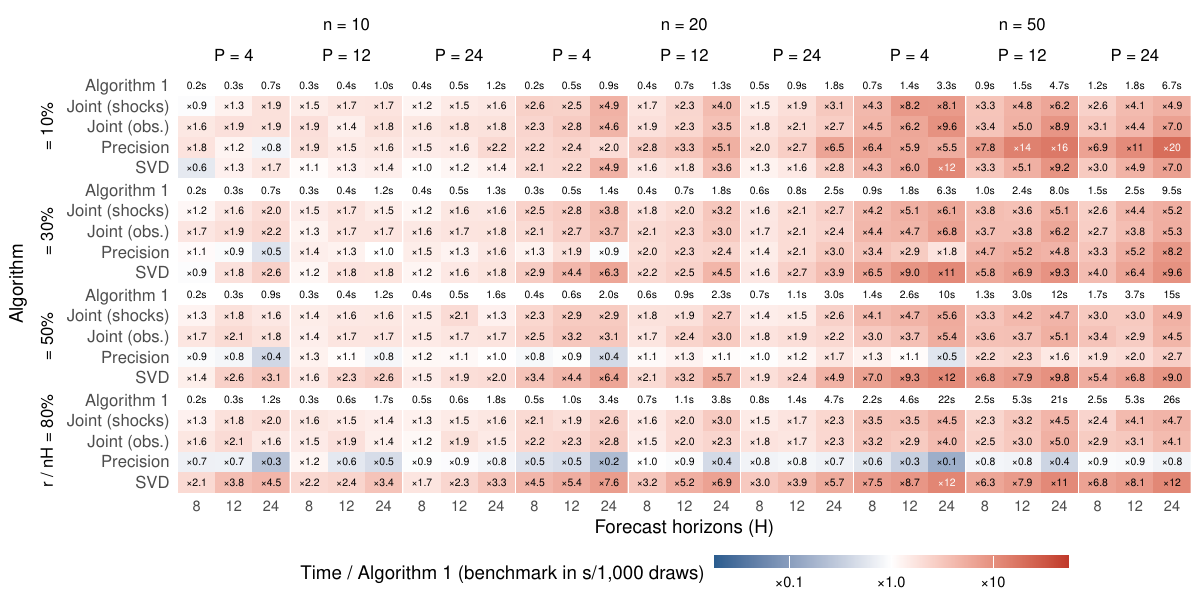}
  \caption{Computation times of competing conditional forecasting algorithms.}
  \label{fig:comp-time}
  \fnote{Timings relative to Algorithm~\ref{alg:xi_post} (red: slower, blue: faster, log scale); the row for Algorithm~\ref{alg:xi_post} reports absolute benchmarks in seconds per 1,000 draws. Panels vary the number of variables $n$ and lag order $P$; the horizontal axis varies the forecast horizon $H$; panel rows vary the restriction density $r/(nH)$.}
\end{figure}

\inlinehead{Computation times} The relative rankings in Figure \ref{fig:comp-time} follow directly from the dominant factorization of each algorithm. The cost of Algorithm~\ref{alg:xi_post} increases with the number of restrictions, whereas that of the precision sampler falls as more of the system is restricted, so the two methods are complements rather than substitutes. For sparse and moderately dense restriction sets, Algorithm~\ref{alg:xi_post} dominates, with the advantage over the precision sampler widening in the size of the system; at $50$ percent density the two methods are roughly on par, and once most of the system is restricted the ranking flips, with the precision sampler fastest in nearly the entire $80$ percent row. The lag order reinforces this pattern. It enters Algorithm~\ref{alg:xi_post} only through the construction of the impulse responses, whereas the bandwidth of the stacked precision matrix grows with $nP$, so that the precision sampler slows down markedly for long lag orders. Overall, the differences across methods are modest in small systems but large at scale: in the biggest sparsely restricted configurations, Algorithm~\ref{alg:xi_post} is faster than the joint samplers by factors of five to ten, the SVD sampler by up to twelve, and the precision sampler by up to twenty.

\section{Empirical Application}\label{sec:application}
We use a dataset patterned after \citet{crump2021large} in our empirical illustration. The quarterly dataset includes $n = 37$ variables for the US, among them three shock proxies described below. After transformations, the data span 1965Q1 to 2025Q4. Linked to the maximum forecast horizon $H = 20$ (i.e., $5$ years until 2030Q4), we use $P = (H - 1) = 19$ lags, which leaves the lag polynomial unrestricted at every horizon of the forecast window, so that the effective estimation sample starts in 1969Q4.\footnote{We pre-screen each series for outliers: observations farther than five interquartile ranges from the central quartiles are treated as missing and imputed prior to estimation, inspired by similar procedures often used with dynamic factor models \citep[see, e.g.,][]{carriero2021addressing}. All series are then standardized to zero mean and unit variance; conditioning targets are mapped into standardized units accordingly, and all results are reported, using an inverse transformation after estimation, in the original units.} All results are based on 80{,}000 sweeps of the Gibbs sampler; we discard the first 20{,}000 as burn-in and retain every 20th of the remaining draws, yielding 3{,}000 posterior draws. A similar dataset has been used in the applications of \citet{arias2026large,chan2025largesign}, from which we also take the sign restrictions to enable a structural interpretation. Specifically, we identify $q = 10$ shocks: demand, investment, financial, monetary policy, government spending, technology, labor supply, wage bargaining, oil price and consumer sentiment. To pin down the loadings of the oil price and monetary policy shocks, we add proxies for them as observables to the VAR, in the spirit of internal instruments: the oil supply news shock of \citet[][available 1975Q1 to 2025Q4]{kanzig2021macroeconomic} loads only on the oil price shock, and the narrative monetary policy shocks of \citet{romer2004new}, as updated by \citet[][1969Q1 to 2007Q4]{wieland2020financial}, and the high-frequency monetary policy shocks of \citet[][1990Q1 to 2025Q4]{jarocinski2020deconstructing} load only on the monetary policy shock. All three proxies are available at monthly frequency, and we sum the monthly values within each quarter. Values outside the coverage of each proxy are treated as missing and sampled as latent states within the Gibbs sampler (OA \ref{app:oa-technical}). Since each proxy loads on its own shock only, the corresponding column of $\bm{L}$ is determined up to sign given $(\bm{L}\bm{L}', \bm{S})$, and the sign restriction on the proxy loading fixes the sign; the conditional forecasts under the DS scheme with the oil price shock as the driving shock therefore do not depend on the prior over the identified set. Detailed variable descriptions and sign restrictions are in Appendix \ref{app:empirical}, and the resulting IRFs for selected key indicators in OA \ref{app:oa-empirical}. The estimated IRFs qualitatively match the responses reported in the cited papers.

\begin{figure}[htbp]
  \centering
  \includegraphics[width=\textwidth]{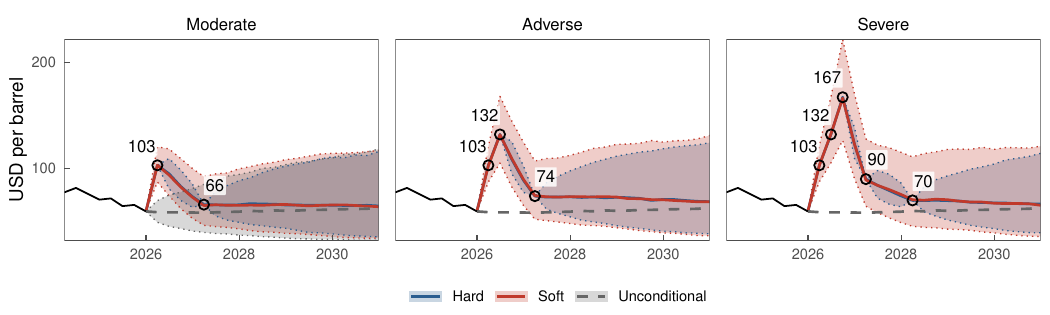}
  \caption{Oil price scenarios: restriction targets and conditional forecasts of the oil price (WTI) under hard and soft constraints.}
  \label{fig:oilscenarios}
  \fnote{Solid black line: observed data; grey dashed line: unconditional forecast (68\% credible set shown in the moderate panel only); circles: restriction targets, labeled with their values in USD per barrel; posterior medians alongside 68\% credible sets.}
\end{figure}

We consider counterfactual scenarios from 2026Q1 onwards, based on assumptions about the nominal price of oil (WTI, in USD per barrel) in the context of the 2026 Iran War and the associated closure of the Strait of Hormuz. The \texttt{moderate}, \texttt{adverse} and \texttt{severe} scenarios are provided by \citet{kilian2026impact} based on a DSGE model described in that paper; they reflect closures of increasing duration. Figure \ref{fig:oilscenarios} shows the timing and the values of the restrictions. All three scenarios share the jump of the oil price on impact; the adverse and severe scenarios rise further for one and two more quarters, respectively, and all three return to a scenario-specific baseline above the pre-war price after one year, the severe scenario with an additional target after two years. At all other horizons the oil price is unrestricted and thus model-determined. All three scenarios describe substantial departures from the unconditional forecast.

Figure \ref{fig:oilscenarios} also shows the resulting predictive distributions when the targets are imposed as hard versus distributional (soft) constraints. For the soft constraints, we impose the restrictions on the mean but retain the model-implied variance, $\bm{\Omega}^{(y)} = \bm{R}^{(y)}\bm{M}'\bm{M}\bm{R}^{(y)\prime}$. For each parameter draw, the conditional means of the two versions coincide exactly, since $\bm{\mu}_\xi$ in (\ref{eq:restrmoments}) does not depend on $\bm{\Omega}$; the tightness of the constraints matters only for the predictive uncertainty, which the soft version carries forward to the unrestricted horizons. In what follows, we focus on the distributional version, so that predictive uncertainty is reflected more adequately \citep[see also the discussions in][]{antolin2021structural}.

The results below combine the scenario restrictions on oil prices, which are identical throughout, with the conditioning schemes of Section \ref{sec:ssa} (RF, ST and DS; OM in OA \ref{app:oa-empirical}). The driving shock under the DS and OM schemes is the structural oil price shock.

\begin{figure}[ht]
  \centering
  \includegraphics[width=\textwidth]{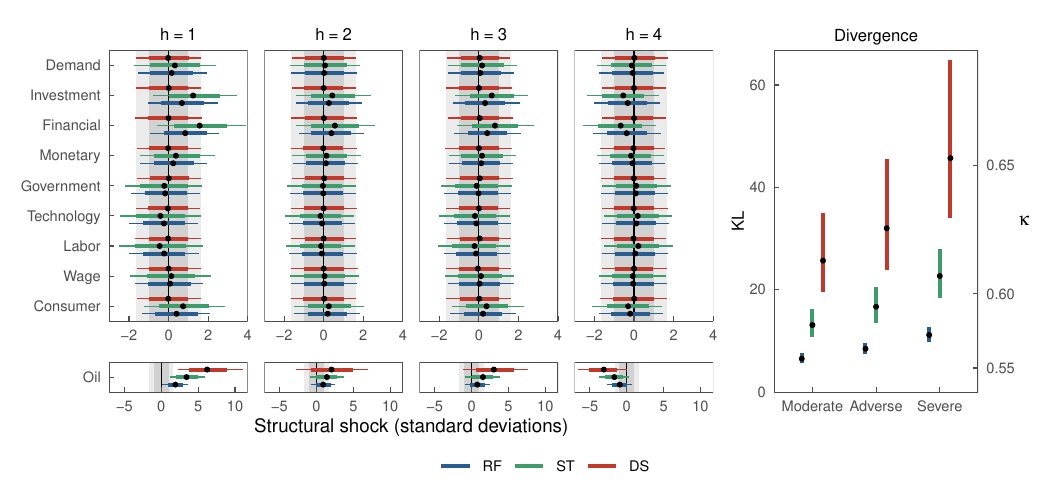}
  \caption{Required structural shocks in the \texttt{severe} scenario (left) and KL divergence and $\kappa$ by scenario (right).}
  \label{fig:eps-scenarios}
  \fnote{Left: medians with $68$/$90$ percent intervals; grey: central $68$/$90$ percent of $\mathcal{N}(0,1)$; oil price shock on its own scale. Right: KL divergence (left axis) and $\kappa$ (right axis); medians, $16$--$84$ percent intervals.}
\end{figure}

The left panel of Figure \ref{fig:eps-scenarios} shows the counterfactual distributions of the structural shocks under the three conditioning schemes against their unconditional distribution, that is, the realizations of the structural shocks required to deliver the stipulated oil price path, for the first four quarters of the severe scenario; OA \ref{app:oa-empirical} reports all scenarios and horizons. The required shocks differ markedly across conditioning schemes, both in their composition and in their magnitudes. Under the RF scheme, the largest contribution on impact comes from the oil price shock itself, at about two standard deviations, but positive financial, investment and consumer sentiment shocks also contribute, alongside small negative technology and labor supply shocks. Reduced-form conditioning thus reads the oil price increase as an oil price shock accompanied by expansionary demand-side shocks, the reading implied by the historical comovement of the oil price with the other series rather than by the supply disruption the scenario describes. While the oil price is still rising, the deviations keep their sign at smaller magnitudes; they reverse in the quarter in which the oil price returns towards its baseline, providing the offsetting forces, and beyond the final restricted horizon no structural shock deviates from its unconditional distribution. A substantial part of the adjustment under the RF scheme is absorbed by the idiosyncratic component of the oil price equation, which deviates by more than the oil price shock on impact (Appendix \ref{app:empirical}), whereas all other idiosyncratic components essentially remain at their unconditional distribution. Under the ST scheme, this margin is closed: the composition of the required shocks is virtually unchanged, but the deviations roughly double, since the structural shocks alone must deliver the scenario.

Once the scenario is attributed to the oil price shock alone (the DS scheme), the reading of a supply disruption, the required magnitude almost doubles again relative to the ST scheme. The median impact shock is about $6$ standard deviations, and nearly identical across the three scenarios since they share the impact target. Further positive realizations follow while the oil price is still rising, and realizations of the opposite sign once it returns towards its baseline. The size of the required shock reflects how much of the oil price variation the oil price shock accounts for; it accounts for about $30$ percent of the innovation variance of the oil price, with the financial and investment shocks the largest of the remaining structural contributors. Delivering the scenario through the oil price shock alone forgoes the contributions of all other innovations, which requires a larger shock.

The right panel of Figure \ref{fig:eps-scenarios} quantifies how far each attribution adjusts the innovations from their unconditional distribution. Under the RF scheme, only the mean part of (\ref{eq:kl}) is non-zero (OA \ref{app:oa-technical}). The divergence rises from $6.5$ to $11.2$ across the moderate to severe scenarios under the RF scheme, to $13.1$ to $22.7$ once the idiosyncratic margin is closed, and to $25.7$ to $45.7$ when the scenario is attributed to the oil price shock alone. In terms of the calibration $\kappa$ (right axis), all three schemes remain between $0.56$ and $0.65$, and the increments are the price of each additional restriction. Closing the idiosyncratic margin costs about as much as the scenario itself under the RF scheme, and attributing the scenario to the oil price shock alone roughly doubles the divergence again, with the increase growing with the severity of the imposed scenario.

\begin{figure}[ht]
  \centering
  \includegraphics[width=\textwidth]{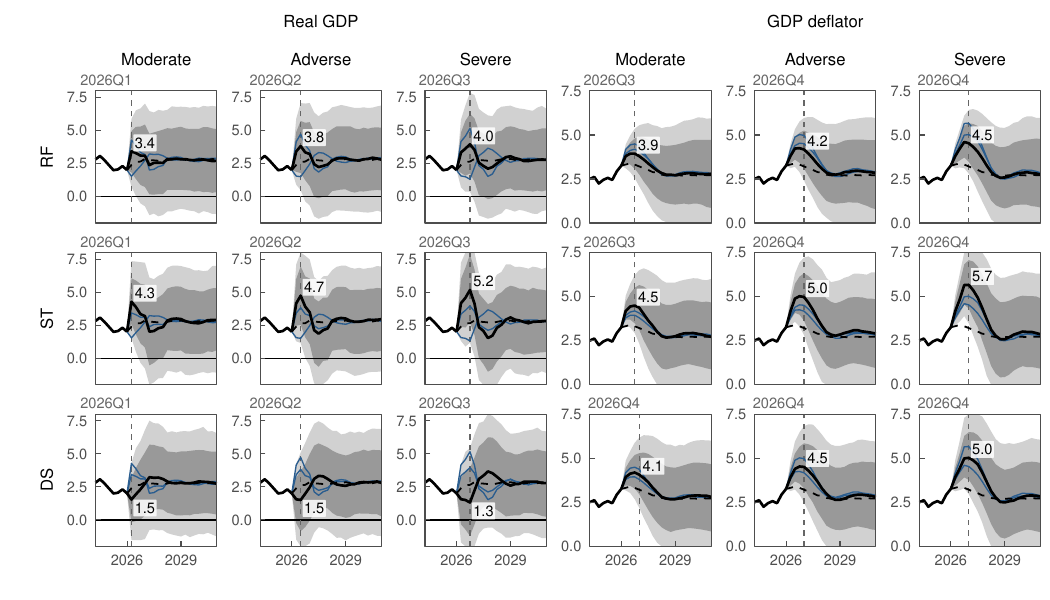}
  \caption{Conditional forecasts of real GDP growth and GDP deflator inflation (year-over-year, in percent) by conditioning scheme (rows) and scenario (columns).}
  \label{fig:cf-rest-types}
  \fnote{RF: reduced-form, ST: structural, DS: driving-shock scheme. Solid black line: observed data; dashed black line: median unconditional forecast. Thick black line and dark and light grey areas: posterior median and 68\% and 90\% credible sets under the scheme of the row; thin blue lines: median conditional forecasts under the other two schemes; y-axis clipped for readability. Dashed vertical line: quarter of the largest median deviation from the unconditional forecast.}
\end{figure}

\begin{figure}[ht]
  \centering
  \includegraphics[width=\textwidth]{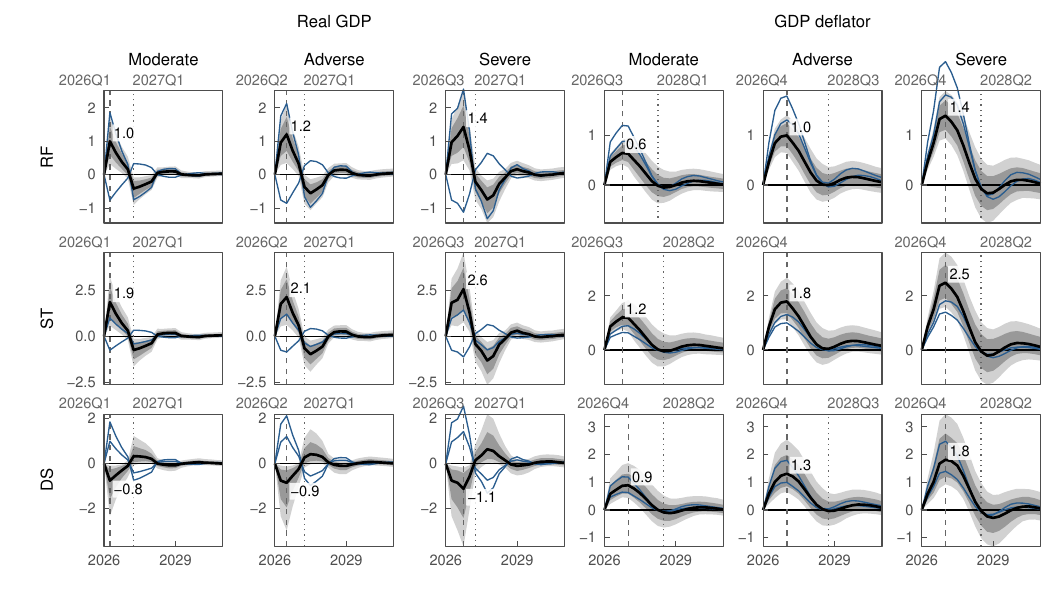}
  \caption{Generalized impulse responses of real GDP growth and GDP deflator inflation (year-over-year, in percentage points) by scheme and scenario.}
  \label{fig:cf-diff}
  \fnote{Differences of the conditional and unconditional forecast means, per posterior draw. Thick black line and dark and light grey areas: posterior median and 68\% and 90\% credible sets under the scheme of the row; thin blue lines: median differences under the other two schemes. Dashed vertical line and label: peak of the median difference (largest in absolute value); dotted vertical line: first subsequent quarter in which the median has returned to zero; dates above each panel. The y-axis scales differ across rows.}
\end{figure}

To economize on space, we focus on the consequences for two unrestricted observables, real GDP growth and inflation in the GDP deflator (both year-over-year). Figure \ref{fig:cf-rest-types} shows their conditional forecasts, and Figure \ref{fig:cf-diff} the differences to the unconditional forecast \citep[generalized impulse response function, GIRF,][]{koop1996impulse}, together with the timing of the peak deviation and of the return to the unconditional path.\footnote{Since the model is linear, the GIRF is the difference between the means of the conditional and the unconditional forecast, $\bm{\mu}_y - \bm{m} = \bm{M}'\bm{\mu}_\xi$, available in closed form for each posterior draw.} Irrespective of the scenario and conditioning scheme, the increase in oil prices is associated with inflationary pressures, but the implied magnitudes differ markedly. Under the RF scheme, median inflation peaks between half a percentage point and one and a half percentage points above the unconditional forecast, depending on the scenario. Closing the idiosyncratic margin strengthens the pass-through. Under the ST scheme, the peak deviations roughly double, reaching $5.7$ percent inflation in the severe scenario. The DS scheme lies between the other two; the expansionary demand-side shocks that accompany the oil price shock under the ST scheme add to inflation, whereas the oil price shock alone raises it by less. Under all schemes, inflation peaks within three quarters of the scenario-specific oil price peak and the median path returns to the unconditional forecast after about two years, with at most a small subsequent undershoot.

The scenario does not result in a contraction of median GDP growth under any of the conditioning schemes, but it shifts the lower tail of the predictive distribution downwards. Growth-at-risk at the 5 percent level \citep{adrian2019vulnerable} falls to between $-1.5$ and $-3.2$ percent across schemes and scenarios, compared with at worst $-1.3$ percent under the unconditional forecast. The shift is largest under the DS scheme in the severe scenario; a contraction has a probability of more than one in four under the DS scheme with the adverse and severe scenarios, and under the ST scheme with the severe scenario. Under the RF and ST schemes, growth is above the unconditional forecast on impact, by one to two and a half percentage points at the peak, slows below it about a year later, and has returned to the unconditional path by mid-2028. Under the DS scheme, the counterfactual paths combine higher inflation with slower growth, but the slowdown is moderate. Expected growth is at most about one percentage point below the unconditional forecast in the quarter of the oil price peak, with median troughs of $1.3$ to $1.5$ percent, and overshoots the unconditional forecast as the oil price returns to its baseline before gradually converging back.

\inlinehead{Reactive monetary policy} The counterfactuals above retain the systematic response of monetary policy embedded in the estimated model. Under the DS scheme, this response is small; the median federal funds rate rises by at most about $25$ basis points above its unconditional forecast (OA \ref{app:oa-empirical}). Our framework also accommodates a deviation from this systematic behavior. We consider policy reactions that raise the federal funds rate by $25$, $50$ or $75$ basis points on impact, relative to its path under the DS scheme, each delivered through a single date-0 monetary policy shock. Since the counterfactual is implemented through the identified policy shock rather than through a change in the policy rule, it is robust to the Lucas critique in the sense of \citet{mckay2023can}, under the assumption that the transmission of the monetary policy shock is invariant to the intervention. With a single shock, it matches the counterfactual reaction on impact only. In terms of the restrictions in Section \ref{sec:econometrics}, the only change relative to the DS scheme is that the monetary policy shock on impact is restricted to $\mathcal{N}(m, 1)$ instead of $\mathcal{N}(0, 1)$, with $m$ chosen for each posterior draw such that the federal funds rate is the stipulated number of basis points above its path under the DS scheme in expectation;\footnote{Because the model is linear, the effects of the policy reaction are proportional to its size and what distinguishes them is the plausibility of the required shock, discussed below.} the oil price shocks over the restricted horizons continue to deliver the oil price path and now also offset the effect of the policy shock on it; this part of the counterfactual is an ordinary conditional forecast, delivered by shocks that are unanticipated in the model, to which the argument of \citet{mckay2023can} does not extend. We report the results for the \texttt{severe} scenario.

\begin{figure}[ht]
  \centering
  \includegraphics[width=\textwidth]{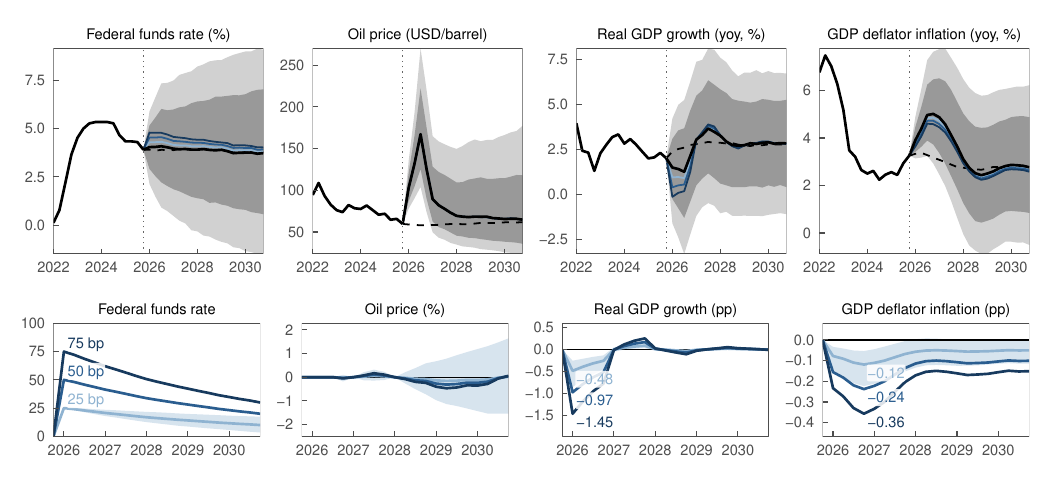}
  \caption{Reactive monetary policy in the \texttt{severe} scenario under the DS scheme: conditional forecasts (top) and effects of the policy reaction (bottom).}
  \label{fig:policy-reaction}
  \fnote{Thick black line and dark and light grey areas: posterior median and 68\% and 90\% credible sets under the systematic policy response; blue lines: posterior medians with a reaction of $25$, $50$ and $75$ basis points on impact (light to dark); dashed line: median unconditional forecast. Bottom: differences between the reactive and the systematic projection.}
\end{figure}
Figure \ref{fig:policy-reaction} shows the forecasts under the systematic response and with reactions of $25$, $50$ and $75$ basis points, together with the effects of the reaction. A $25$ basis point reaction lowers inflation by about a tenth of a percentage point on impact, with a peak effect of $0.12$ percentage points after three quarters that fades only gradually, as the federal funds rate remains elevated. The cost in terms of output is front-loaded: year-over-year GDP growth is about half a percentage point lower on impact, and the effect vanishes after four quarters. Relative to the scenario itself, these effects are modest. Under the systematic response, median inflation peaks at $5$ percent, about two percentage points above the unconditional forecast; even a $75$ basis point reaction lowers this peak by less than half a percentage point, at the cost of GDP growth slightly below zero on impact. Table \ref{tab:policy-reaction} in Appendix \ref{app:empirical} assesses the plausibility of reactions of different sizes by the required monetary policy shock, its effect on the divergence, and the historical frequency of shocks of that size. A $25$ basis point reaction requires a shock of $1.6$ standard deviations and barely changes the divergence; shocks of this size occur in $7$ percent of the quarters in our sample, but in only $3$ percent of the quarters since 1990. Reactions of $50$ and $75$ basis points require shocks of about $3$ and $5$ standard deviations, which raise the divergence noticeably and have no precedent in the sample since 1990. Within the model, a sizable deviation from the systematic response is thus difficult to reconcile with the historical conduct of monetary policy.

\section{Conclusions}\label{sec:conclusions}
This paper develops methods for SSA in high-dimensional BVARs with a factor structure on the reduced-form errors. The framework accommodates separate distributional restrictions on observables, structural shocks and idiosyncratic components, and we propose a fast simulation-based algorithm. In an application to oil price scenarios in the context of the 2026 closure of the Strait of Hormuz, we show that the attribution of a scenario shapes the counterfactual predictions. Attributed to the structural oil price shock alone, the stipulated oil price path implies higher inflation and a moderate slowdown, with median growth remaining positive; with its structural origins unrestricted, the model reads it as an oil price shock accompanied by a demand boom. The divergence of the required shocks from their unconditional distribution measures what each attribution costs. The same machinery delivers a policy counterfactual in which monetary policy reacts to varying degrees to the oil price increase.

The analysis is linear and homoskedastic and rests on a single identification scheme. Future research avenues include extensions of the model to include stochastic volatility and other forms of parameter time variation, which would allow predictive uncertainty and transmission dynamics to change over time; conditional on the respective parameter paths, the model remains Gaussian and the algorithms developed in this paper apply with minor modifications. Matching entire counterfactual paths of the policy instrument, rather than its reaction on impact, requires combining several identified policy shocks, e.g., news about future policy, as in \citet{mckay2023can}. This also relates to the approach of \citet{caravello2024evaluating}, whose counterfactuals build on two inputs, namely conditional forecasts and the causal effects of policy shocks, which correspond to $\bm{m}$ and blocks of $\bm{M}'$ in our framework.

{\setstretch{1.0}\putbib}
\end{bibunit}

\FloatBarrier
\normalsize\clearpage\doublespacing

%--- Appendix
\setcounter{section}{0}
\setcounter{equation}{0}
\setcounter{figure}{0}
\setcounter{table}{0}
\renewcommand\thesection{\Alph{section}}
\renewcommand\theequation{\Alph{section}.\arabic{equation}}
\renewcommand\thetable{\Alph{section}.\arabic{table}}
\renewcommand\thefigure{\Alph{section}.\arabic{figure}}
\counterwithin*{equation}{section}
\counterwithin*{figure}{section}
\counterwithin*{table}{section}

\section{Technical Appendix}\label{app:technical}
\subsection{Static representation of the joint forecast vector}\label{app:staticrep}
The matrix $\bm{H} = (\bm{I}_{nH} - \tilde{\bm{H}})$ encodes the dynamic coefficients while $\bm{h} = \tilde{\bm{h}} + (\bm{\iota}_{H}\otimes\bm{a})$ collects the predetermined terms (initial conditions and intercepts). Both are composed of $n\times n$ blocks indexed by horizon. Write $[\,\cdot\,]_{ij}$ for the $(i,j)$th $n\times n$ block of a matrix with $H\times H$ such blocks, $[\,\cdot\,]_{i}$ for the $i$th $n\times 1$ block of a vector of length $nH$, and set $\bm{A}_p = \bm{0}_{n\times n}$ for $p > P$. Then, for $i,j = 1,\hdots,H$,
\begin{equation*}
[\tilde{\bm{h}}]_{i} = \sum_{p=i}^{P} \bm{A}_p\,\bm{y}_{T+i-p},
\qquad
[\tilde{\bm{H}}]_{ij} = \begin{cases} \bm{A}_{i-j} & i > j,\\ \bm{0}_{n\times n} & i \le j,\end{cases}
\end{equation*}
where the sum defining $[\tilde{\bm{h}}]_{i}$ is empty for $i > P$. That is, $\tilde{\bm{h}}$ collects the contributions of the observed data to each forecast horizon, and $\tilde{\bm{H}}$ is block lower-triangular with the $p$th block subdiagonal consisting of copies of $\bm{A}_p$; the zero-padding covers $H < P$. The inverse of $\bm{H} = \bm{I}_{nH} - \tilde{\bm{H}}$ has the VMA coefficients in place of the autoregressive ones:
\begin{equation*}
[\bm{H}^{-1}]_{ij} = \begin{cases} \bm{\Phi}_{i-j} & i \ge j,\\ \bm{0}_{n\times n} & i < j,\end{cases}, \quad \bm{\Phi}_0 = \bm{I}_n.
\end{equation*}

Since $\bm{H}$ is lower block-triangular with $\bm{I}_n$ on its main block diagonal, $\det(\bm{H}) = \prod_{i=1}^{H} \det(\bm{I}_n) = 1$, and thus invertible. To establish the structure of $\bm{H}^{-1}$, solve $\bm{H}\bm{x} = \bm{E}_j$ block by block via forward substitution, where $\bm{E}_j = [\bm{0}_{n\times n(j-1)}, \bm{I}_n, \bm{0}_{n\times n(H-j)}]'$ is an $nH \times n$ matrix. For $i = 1$ the sum is empty and the equation reduces to $\bm{x}_1 = \mathbb{I}(1 = j)\cdot\bm{I}_n$. For $i > 1$, the $i$th block equation reads $\bm{x}_i - \sum_{p=1}^{\min(P,\,i-1)} \bm{A}_p\,\bm{x}_{i-p} = \mathbb{I}(i = j)\cdot\bm{I}_n$.

For $i < j$, forward substitution gives $\bm{x}_i = \bm{0}_n$. For $i = j$, $\bm{x}_j = \bm{I}_n = \bm{\Phi}_0$. For $i > j$, the recursion $\bm{x}_i = \sum_{p=1}^{\min(P,\,i-j)} \bm{A}_p\,\bm{x}_{i-p}$ coincides with the VMA recursion $\bm{\Phi}_s = \sum_{p=1}^{\min(P,s)} \bm{A}_p\,\bm{\Phi}_{s-p}$, so that $\bm{x}_i = \bm{\Phi}_{i-j}$ follows by induction. The columns of $\bm{H}^{-1}$ are thus given by the solutions $\bm{x}$. In practice, our code never explicitly forms $\bm{H}^{-1}$ or the full matrix $\bm{M}'$ but constructs them recursively, which is algebraically equivalent and offers substantial computational gains.

\subsection{Proof of Proposition \ref{prop:schemes}}\label{app:proof-schemes}
\begin{proof}
Let $\tilde{\bm{L}} = \bm{L}\bm{Q}$ be a rotation of the loadings and mark with a tilde every quantity computed from $\tilde{\bm{L}}$ in place of $\bm{L}$. With the orthogonal $d \times d$ matrix $\bm{\mathcal{Q}} = \bdiag(\bm{I}_H \otimes \bm{Q}, \bm{I}_{nH})$, we have $\tilde{\bm{\mathcal{V}}} = \bm{\mathcal{V}}\bm{\mathcal{Q}}$ and $\tilde{\bm{M}}' = \bm{M}'\bm{\mathcal{Q}}$, whereas $\tilde{\bm{m}} = \bm{m}$ and $\tilde{\bm{M}}'\tilde{\bm{M}} = \bm{M}'\bm{M} = \bm{H}^{-1}\bm{\mathcal{V}}\bm{\mathcal{V}}'\bm{H}^{-1\prime} = \bm{H}^{-1}(\bm{I}_H \otimes \bm{\Sigma})\bm{H}^{-1\prime}$, which depends on $(\bm{L}, \bm{S})$ only through $\bm{\Sigma}$.
\begin{enumerate}
    \item[(i)] Under the RF scheme, $\bm{R} = \bm{R}^{(y)}\bm{M}'$, so $\bm{R}\bm{R}' = \bm{R}^{(y)}\bm{M}'\bm{M}\bm{R}^{(y)\prime}$ and $\bm{M}'\bm{R}' = \bm{M}'\bm{M}\bm{R}^{(y)\prime}$. Inserting (\ref{eq:restrmoments}) into $\bm{\mu}_y = \bm{m} + \bm{M}'\bm{\mu}_\xi$ and $\bm{V}_y = \bm{M}'\bm{V}_\xi\bm{M}$, the matrix $\bm{M}'$ enters only through these two products, so both moments are functions of $\bm{m}$, $\bm{M}'\bm{M}$ and the targets.
    \item[(ii)] Under the ST scheme, $\bm{R} = [\bm{R}^{(y)}\bm{M}'; \bm{0}_{r_u \times qH}, \bm{R}^{(u)}]$. The second block is zero in the structural columns, so $\tilde{\bm{R}} = \bm{R}\bm{\mathcal{Q}}$, while $\bm{r}$, $\bm{\Omega}$ and $\tilde{\bm{R}}\tilde{\bm{R}}' = \bm{R}\bm{R}'$ are unchanged. By (\ref{eq:restrmoments}), $\tilde{\bm{\mu}}_\xi = \bm{\mathcal{Q}}'\bm{\mu}_\xi$ and $\tilde{\bm{V}}_\xi = \bm{\mathcal{Q}}'\bm{V}_\xi\bm{\mathcal{Q}}$, so the restricted innovations under $\tilde{\bm{L}}$ are $\tilde{\bm{\xi}}^{\ast} = \bm{\mathcal{Q}}'\bm{\xi}^{\ast}$ and the restricted structural shocks $\tilde{\bm{\varepsilon}}^{\ast} = (\bm{I}_H \otimes \bm{Q}')\bm{\varepsilon}^{\ast}$. Since $\tilde{\bm{M}}'\bm{\mathcal{Q}}' = \bm{M}'$, it follows that $\tilde{\bm{\mu}}_y = \bm{\mu}_y$ and $\tilde{\bm{V}}_y = \bm{V}_y$. The split enters through the block $\bm{H}^{-1}\bm{\mathcal{S}}^{1/2}_{(H)}\bm{R}^{(u)\prime}$ of $\bm{M}'\bm{R}'$, so given $\bm{\Sigma}$, $\bm{y}^{\ast}$ depends on the decomposition $(\bm{L}\bm{L}', \bm{S})$.
    \item[(iii)] Under the OM and DS schemes, $\bm{R}$ contains the block $[\bm{R}^{(\varepsilon)}, \bm{0}_{r_\varepsilon \times nH}]$, which is not zero in the structural columns: $[\bm{R}^{(\varepsilon)}, \bm{0}]\bm{\mathcal{Q}} = [\bm{R}^{(\varepsilon)}(\bm{I}_H \otimes \bm{Q}), \bm{0}]$ differs from $[\bm{R}^{(\varepsilon)}, \bm{0}]$ unless $\bm{Q}$ leaves the restricted shocks unchanged, so $\tilde{\bm{R}} \neq \bm{R}\bm{\mathcal{Q}}$ in general.
\end{enumerate}
\end{proof}

\section{Empirical Appendix}\label{app:empirical}
Table \ref{tab:variables} lists the variables used in the empirical application, series codes (FRED-QD where available; constructed series are described in the table notes), and transformations. The last three rows are the shock proxies described in Section \ref{sec:application}. Most variables enter in year-over-year growth rates except those already in percent; the unemployment rate and capacity utilization enter in first differences. Note that the oil price enters in log-levels, because the scenarios are stated in terms of the nominal price of oil and encoding them for a differenced series would be more complicated.

\begin{table}[ht]
  \centering
  \caption{Reactive monetary policy.}
  \label{tab:policy-reaction}
  \small
  % Auto-generated by out_tables_results.R -- do not edit by hand.
% Reactive monetary policy (proxy fit hom-out_ARRSBP_tr0_s2base_chain2_OILNEWS-MPRR-MPJK): mean of the date-0 Monetary shock that raises the fed funds rate by the given amount on impact relative to systematic policy, median [16%, 84%] over draws; KL and calibrated kappa increase over the DS scheme (severe scenario), median; share of quarters with an estimated Monetary shock at least as large in absolute value. In 1.1% of draws the Monetary shock barely moves the funds rate on impact (extreme required shocks); medians only.
\begin{tabular}{lrrrrr}
\toprule
 & & \multicolumn{2}{c}{Change in} & \multicolumn{2}{c}{Frequency (\%)} \\
\cmidrule(lr){3-4} \cmidrule(lr){5-6}
Reaction & Shock (s.d.) & KL & $\kappa$ & 1969--2025 & Since 1990 \\
\midrule
+25 bp & 1.56 [1.20, 2.22] & 0.2 & 0.000 & 7.0 & 2.6 \\
+50 bp & 3.11 [2.41, 4.45] & 3.1 & 0.005 & 1.0 & 0.0 \\
+75 bp & 4.67 [3.61, 6.67] & 8.4 & 0.013 & 0.3 & 0.0 \\
\bottomrule
\end{tabular}

  \fnote{Required monetary policy shock on impact (posterior median alongside 68\% credible set). Change in: median change in the KL divergence and $\kappa$ relative to the driving-shock scheme. Frequency: share of quarters with an estimated monetary policy shock at least as large.}
\end{table}

\begin{figure}[ht]
  \centering
  \includegraphics[width=\textwidth]{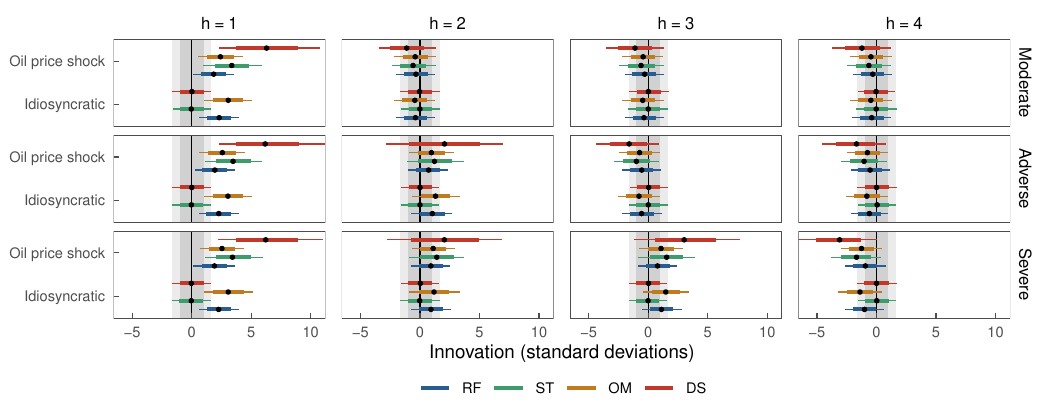}
  \caption{Structural oil price shock and idiosyncratic component of the oil price equation by horizon, scenario and conditioning scheme.}
  \label{fig:open-margin}
  \fnote{OM: open-margin scheme. Medians with $68$/$90$ percent intervals; grey: central $68$/$90$ percent of $\mathcal{N}(0,1)$.}
\end{figure}

% Auto-generated by out_tables.R -- do not edit by hand.
% Sources: prep_data.R (vinfo), set_rest.R (L_sign_ID).

\begin{table}[htbp]
  \centering
  \caption{Variables, transformations, and sign/ranking restrictions}
  \label{tab:variables}
  \scriptsize
  \begin{threeparttable}
  \begin{tabular*}{\textwidth}{@{\extracolsep{\fill}} l l *{10}{c}}
    \toprule
    \textbf{Code} & \textbf{Description (Tr.)} & (\textbf{a}) & (\textbf{b}) & (\textbf{c}) & (\textbf{d}) & (\textbf{e}) & (\textbf{f}) & (\textbf{g}) & (\textbf{h}) & (\textbf{i}) & (\textbf{j}) \\
    \midrule
    \texttt{GDPC1} & Real GDP (1) & $+$ & $+$ & $+$ & $-$ & $+$ & $+$ & $+$ & $+$ & $-$ & $+$ \\
    \texttt{PCECC96} & Real PCE (1) & $+$ &  &  &  &  & $+$ &  &  & $-$ & $+$ \\
    \texttt{PRFIx} & Residential investment (1) &  & $+$ &  &  &  &  &  &  &  & $+$ \\
    \texttt{PNFIx} & Nonresidential investment (1) &  & $+$ &  &  &  & $+$ &  &  & $-$ & $+$ \\
    \texttt{EXPGSC1} & Exports (1) &  &  &  &  &  &  &  &  &  &  \\
    \texttt{IMPGSC1} & Imports (1) &  &  &  &  &  &  &  &  &  &  \\
    \texttt{GCEC1} & Govt consumption \& inv. (1) &  &  &  &  & $+$ &  &  &  &  &  \\
    \texttt{GDPCTPI} & GDP deflator (1) & $+$ & $+$ & $+$ & $-$ & $+$ & $-$ & $-$ & $-$ & $+$ & $+$ \\
    \texttt{PCECTPI} & PCE deflator (1) & $+$ & $+$ & $+$ & $-$ & $+$ & $-$ & $-$ & $-$ & $+$ & $+$ \\
    \texttt{PCEPILFE} & Core PCE deflator (1) & $+$ & $+$ & $+$ & $-$ & $+$ & $-$ & $-$ & $-$ & $+$ & $+$ \\
    \texttt{CPIAUCSL} & CPI (1) & $+$ & $+$ & $+$ & $-$ & $+$ & $-$ & $-$ & $-$ & $+$ & $+$ \\
    \texttt{CPILFESL} & Core CPI (1) & $+$ & $+$ & $+$ & $-$ & $+$ & $-$ & $-$ & $-$ & $+$ & $+$ \\
    \texttt{AHETPIx} & Avg hourly earnings (1) &  &  &  &  &  & $+$ & $-$ & $-$ & $-$ &  \\
    \texttt{OPHNFB} & Labor productivity (1) &  &  &  &  &  & $+$ &  &  & $-$ &  \\
    \texttt{PAYEMS} & Nonfarm payrolls (1) &  &  &  & $-$ &  &  &  &  &  &  \\
    \texttt{UNRATE} & Unemployment rate (2) & $-$ & $-$ & $-$ & $+$ & $-$ & $-$ & $+$ & $-$ & $+$ & $+$ \\
    \texttt{INDPRO} & Industrial production (1) & $+$ & $+$ & $+$ & $-$ &  &  &  &  &  &  \\
    \texttt{CUMFNS} & Capacity utilization (2) & $+$ & $+$ & $+$ & $-$ &  &  &  &  &  &  \\
    \texttt{TFP} & Util.-adjusted TFP (1) &  &  &  &  &  & $+$ &  &  & $0$ &  \\
    \texttt{HOUST} & Housing starts (1) &  &  &  &  &  &  &  &  &  &  \\
    \texttt{DPIC96} & Real disposable income (1) &  &  &  &  &  &  &  &  &  &  \\
    \texttt{UMCSENTx} & Consumer sentiment (3) &  &  &  &  &  &  &  &  &  &  \\
    \texttt{FEDFUNDS} & Federal funds rate (3) & $+$ & $+$ & $+$ & $+$ & $+$ &  &  &  &  &  \\
    \texttt{TB3MS} & 3-month T-bill rate (3) & $+$ & $+$ & $+$ & $+$ & $+$ &  &  &  &  &  \\
    \texttt{GS5} & 5-year Treasury note (3) &  &  &  & $+$ &  &  &  &  &  &  \\
    \texttt{GS10} & 10-year Treasury bond (3) &  &  &  & $+$ &  &  &  &  &  &  \\
    \texttt{AAASPR} & AAA spread (3) &  &  &  & $+$ &  &  &  &  &  &  \\
    \texttt{BAASPR} & BAA spread (3) &  &  &  & $+$ &  &  &  &  &  &  \\
    \texttt{MPRIME} & Prime rate (3) & $+$ & $+$ & $+$ & $+$ & $+$ &  &  &  &  &  \\
    \texttt{SP500} & S\&P 500 (1) &  & $-$ & $+$ & $-$ &  &  &  &  &  & $+$ \\
    \texttt{EXUSUKx} & USD/GBP exchange rate (1) &  &  &  &  &  &  &  &  &  &  \\
    \texttt{OILPRICEx} & Oil price (4) &  &  &  &  &  &  &  &  & $+$ &  \\
    \texttt{PNFI\_GDP} & Nonres. inv./GDP (3) & $-$ & $+$ & $+$ &  &  &  &  &  &  &  \\
    \texttt{GCEC\_GDP} & Govt spending/GDP (3) & $-$ & $-$ & $-$ &  & $+$ &  &  &  &  &  \\
    \texttt{OILNEWS} & Oil supply news (K) (3) & $0$ & $0$ & $0$ & $0$ & $0$ & $0$ & $0$ & $0$ & $+$ & $0$ \\
    \texttt{MPRR} & Monetary policy (RR) (3) & $0$ & $0$ & $0$ & $+$ & $0$ & $0$ & $0$ & $0$ & $0$ & $0$ \\
    \texttt{MPJK} & Monetary policy (JK) (3) & $0$ & $0$ & $0$ & $+$ & $0$ & $0$ & $0$ & $0$ & $0$ & $0$ \\
    \bottomrule
  \end{tabular*}
  \begin{tablenotes}[flushleft]
    \tiny
    \item[] $+$ ($-$): positive (negative) sign of the loading, i.e., of the impact response; $0$: loading fixed at zero; empty: unrestricted. \texttt{PNFI\_GDP} and \texttt{GCEC\_GDP} ($100\times$ ratios to real GDP) carry ranking restrictions: $+$ ($-$) requires the numerator to respond more (less) strongly than GDP. \texttt{AAASPR}/\texttt{BAASPR}: Moody's AAA/BAA yields minus the federal funds rate; \texttt{TFP}: Fernald's utilization-adjusted TFP. External proxies, summed within the quarter, missing outside their coverage (sampled as latent states), each loading only on its own shock: K = \citet{kanzig2021macroeconomic}; RR = \citet{romer2004new}, updated by \citet{wieland2020financial}; JK = \citet{jarocinski2020deconstructing}. All other codes are FRED-QD mnemonics. \textit{Transformations:} (1)~$100\,\Delta_4\log x$; (2)~$\Delta x$; (3)~level; (4)~$100\log x$. \textit{Shocks:} (\textbf{a})~Demand; (\textbf{b})~Investment; (\textbf{c})~Financial; (\textbf{d})~Monetary policy; (\textbf{e})~Government spending; (\textbf{f})~Technology; (\textbf{g})~Labor supply; (\textbf{h})~Wage bargaining; (\textbf{i})~Oil price; (\textbf{j})~Consumer sentiment.
  \end{tablenotes}
  \end{threeparttable}
\end{table}

Figure \ref{fig:open-margin} compares the two margins that can deliver the stipulated oil price path: the structural oil price shock and the idiosyncratic component of the oil price equation. Under the schemes that leave the idiosyncratic margin open, the idiosyncratic component absorbs at least as much of the scenario on impact as the oil price shock: its median is $2.3$ standard deviations under the RF and about $3$ under the OM scheme, against oil price shocks of $1.8$ to $1.9$ and $2.4$ to $2.6$ standard deviations. It also follows the subsequent course of the oil price, with medians of about $0.9$ to $1.5$ standard deviations at $h = 2, 3$ and $-1.0$ to $-1.4$ at $h = 4$ in the severe scenario. Once the margin is closed (ST and DS), the component remains at its unconditional distribution and the oil price shock takes over, with $3.4$ to $3.5$ and about $6$ standard deviations on impact.

\FloatBarrier
\normalsize\clearpage\doublespacing

%--- Online Appendix
\begin{center}
{\LARGE\sffamily\textbf{Online Appendix:\\\huge\titletext}}
\end{center}

\setcounter{page}{1}
\renewcommand{\thepage}{S\arabic{page}}
\setcounter{section}{2} % continues the appendix lettering (C, D)
\setcounter{equation}{0}
\setcounter{figure}{0}
\setcounter{table}{0}
\setcounter{footnote}{0}
\renewcommand\thesection{\Alph{section}}
\renewcommand\theequation{\thesection.\arabic{equation}}
\renewcommand\thetable{\thesection.\arabic{table}}
\renewcommand\thefigure{\thesection.\arabic{figure}}

\begin{bibunit}
\section{Technical Appendix}\label{app:oa-technical}
\subsection{Priors, posteriors and computational aspects of the algorithm}\label{app:priorsposteriors}
Models as in (\ref{eq:model-main}) and (\ref{eq:model-error}) have several computational advantages. Using $\bm{y}_t - \bm{L}\bm{\varepsilon}_t = \tilde{\bm{y}}_t = (\tilde{y}_{1t},\hdots,\tilde{y}_{nt})'$, $\bm{A}_{i\bullet}$ is the $i$th row of $\bm{A}$, $\tilde{\bm{x}}_t = (\bm{x}_t',1)'$ and $\bm{\beta}_i = (\bm{A}_{i\bullet}, a_i)'$, we obtain the $i$th VAR equation as: 
\begin{equation*}
    \tilde{y}_{it} = \tilde{\bm{x}}_t'\bm{\beta}_i + s_i u_{it}, \quad u_{it} \sim \mathcal{N}(0,1), \quad i = 1,\hdots,n.
\end{equation*}

We rely on a horseshoe prior, which is among the class of global--local priors. This prior assumes vec$(\bm{A})\sim\mathcal{N}(\bm{0}_k, \tau^2\cdot\diag(\lambda_1^2,\hdots,\lambda_k^2))$. The global parameter $\tau\sim\mathcal{C}^{+}(0,1)$, where $\mathcal{C}^{+}$ is the half-Cauchy distribution, controls the overall amount of shrinkage towards the prior mean, which may be offset by the local scalings $\lambda_j\sim\mathcal{C}^{+}(0,1)$ for $j = 1,\hdots,k$, when the corresponding predictor is important. The resulting prior is $\bm{\beta}_i\sim\mathcal{N}(\bm{0}, \underline{\bm{V}}_{\bm{\beta} i})$, where the diagonal elements of $\underline{\bm{V}}_{\bm{\beta} i}$ are the products of the squared global and local scales ($\tau^2\lambda_j^2$) associated with the $i$th equation, and we add a normal prior with variance $10$ for the intercept. The hyperparameters $\tau$ and $\lambda_j$ are updated with auxiliary inverse-Gamma variables \citep{makalic2015simple}. Draws from these posteriors can be updated equation-by-equation with standard formulas, using $\underline{\bm{y}}_i = (\tilde{y}_{i1},\hdots,\tilde{y}_{iT})'/s_{i}$ and $\underline{\bm{X}}_i = (\tilde{\bm{x}}_1,\hdots,\tilde{\bm{x}}_T)'/s_{i}$, for a conditionally conjugate Gaussian linear regression:
\begin{equation}
    \bm{\beta}_i|\bullet \sim \mathcal{N}(\overline{\bm{V}}_{\bm{\beta}i}\,\underline{\bm{X}}_i'\underline{\bm{y}}_i,\;\overline{\bm{V}}_{\bm{\beta}i}), \quad \overline{\bm{V}}_{\bm{\beta}i} = (\underline{\bm{V}}_{\bm{\beta}i}^{-1} + \underline{\bm{X}}_i'\underline{\bm{X}}_i)^{-1}.\label{eq:VARparas}
\end{equation}
When $nP \gg T$ we follow \citet{kastner2020sparse} and use the fast sampling algorithm of \citet{bhattacharya2016fast} which significantly speeds up the algorithm. A draw is obtained by first sampling $\bm{c}\sim\mathcal{N}(\bm{0}_{nP+1},\underline{\bm{V}}_{\bm{\beta} i})$ and $\bm{d}\sim\mathcal{N}(\bm{0}_{T},\bm{I}_{T})$. Set $\bm{v} = \underline{\bm{X}}_i\bm{c} + \bm{d}$ and compute $\bm{w} = (\underline{\bm{X}}_i\underline{\bm{V}}_{\bm{\beta} i}\underline{\bm{X}}_i' + \bm{I}_T)^{-1}(\underline{\bm{y}}_i - \bm{v})$. A draw from the posterior of $\bm{\beta}_i$ is obtained as $\bm{c} + \underline{\bm{V}}_{\bm{\beta} i}\underline{\bm{X}}_i'\bm{w}$. This ties the respective dominant computational cost to the sample size rather than the number of parameters.

We denote by $\bm{l}_i'$ the $i$th row and $l_{ij}$ the $(i,j)$th element of the matrix $\bm{L}$ with $i = 1, \hdots, n$ and $j = 1, \hdots, q$. The prior on the loadings implements the restrictions $\bm{\mathcal{R}}$ introduced in Section \ref{sec:econometrics}:
\begin{equation*}
    {l}_{ij}\sim
    \begin{cases}
	    \mathcal{N}(m_{L,ij}, v_{L,ij})\cdot\mathbb{I}({l}_{ij} \geq 0), &\text{if } \mathcal{R}_{ij} = +,\\
	    \mathcal{N}(m_{L,ij}, v_{L,ij})\cdot\mathbb{I}(l_{ij}\leq0), &\text{if } \mathcal{R}_{ij} = -,\\
	    \delta_{c_{ij}}, &\text{if } \mathcal{R}_{ij} = c_{ij},\\
	    \mathcal{N}(m_{L,ij}, v_{L,ij}),&\text{if } \mathcal{R}_{ij} = \text{NA},
    \end{cases}
\end{equation*}
where $\delta_{c_{ij}}$ is the point mass at $c_{ij}$; we use weakly informative choices for the prior means $m_{L,ij} = 0$ and variances $v_{L,ij} = 3$. The elements of the lower ($\underline{\bm{l}}_i$) and upper ($\overline{\bm{l}}_i$) bounds are set as follows: if $\mathcal{R}_{ij} = +\rightarrow(0,\infty)$, if $\mathcal{R}_{ij} = -\rightarrow(-\infty,0)$, if $\mathcal{R}_{ij} = \text{NA}\rightarrow(-\infty,\infty)$. Point-mass restrictions $\mathcal{R}_{ij} = c_{ij}$ fix the respective loadings and exclude them from the update (for $c_{ij} = 0$, this is a zero restriction); let $\mathcal{F}_i$ collect the free indices of row $i$ and $\bm{l}_{i,\mathcal{F}_i}$ the corresponding subvector, with the bounds above referring to the free coordinates. Stacking the reduced form error $\eta_{it} = \bm{l}_i'\bm{\varepsilon}_t + s_i u_{it}$ across $t = 1,\hdots,T$, with $\bm{E} = (\bm{\varepsilon}_1,\hdots,\bm{\varepsilon}_T)'$, yields $\bm{\eta}_i = \bm{E}\bm{l}_i + s_i\bm{u}_{i}$. Collecting in $\bm{E}_{\mathcal{F}_i}$ the columns of $\bm{E}$ associated with the free loadings, and adjusting for the contribution of the fixed ones via $\bm{\eta}_i^{c} = \bm{\eta}_i - \bm{E}_{\mathcal{F}_i^{c}}\bm{l}_{i,\mathcal{F}_i^{c}}$ (with $\mathcal{F}_i^{c}$ the fixed indices), using ${\bm{V}_{\bm{l}_i}} = (s_i^{-2}\bm{E}_{\mathcal{F}_i}'\bm{E}_{\mathcal{F}_i} + \underline{\bm{V}}_{L,i}^{-1})^{-1}$ the posterior of the free equation-specific loadings is truncated normal:
\begin{equation}
    \bm{l}_{i,\mathcal{F}_i}|\bullet \sim \mathcal{N}\big(\bm{V}_{\bm{l}_i}(s_i^{-2}\bm{E}_{\mathcal{F}_i}'\bm{\eta}_i^{c} + \underline{\bm{V}}_{L,i}^{-1}\bm{m}_{L,i}),\,{\bm{V}_{\bm{l}_i}}\big)\cdot\mathbb{I}(\underline{\bm{l}}_i < \bm{l}_{i,\mathcal{F}_i} < \overline{\bm{l}}_i),
\end{equation}
where $\bm{m}_{L,i}$ and the diagonal matrix $\underline{\bm{V}}_{L,i}$ collect the prior means and variances of the free loadings. We sample from this multivariate truncated normal distribution using the minimax tilting method of \citet{botev2017normal}.

A sample from the posterior of the factors can be obtained by writing the model in a seemingly unrelated regression form, using $\bm{\varepsilon}_{(1:T)} = (\bm{\varepsilon}_1',\hdots,\bm{\varepsilon}_{T}')'$, $\bm{\eta}_{(1:T)} = (\bm{\eta}_1',\hdots,\bm{\eta}_T')'$, $\bm{u}_{(1:T)} = (\bm{u}_1',\hdots,\bm{u}_T')'$, $\bm{\mathcal{L}}_{(T)} = (\bm{I}_T\otimes \bm{L})$, $\bm{\mathcal{S}}_{(T)}^{1/2} = (\bm{I}_T \otimes \bm{S}^{1/2})$, so that $\bm{\eta}_{(1:T)} = \bm{\mathcal{L}}_{(T)}\bm{\varepsilon}_{(1:T)} + \bm{\mathcal{S}}_{(T)}^{1/2}\bm{u}_{(1:T)}$. The posterior is a normal distribution:
\begin{equation}
    \bm{\varepsilon}_{(1:T)}|\bullet \sim \mathcal{N}(\bm{V}_{\varepsilon}\bm{\mathcal{L}}_{(T)}'\bm{\mathcal{S}}_{(T)}^{-1}\bm{\eta}_{(1:T)}, \bm{V}_{\varepsilon}), \quad \bm{V}_{\varepsilon} = (\bm{\mathcal{L}}_{(T)}'\bm{\mathcal{S}}_{(T)}^{-1}\bm{\mathcal{L}}_{(T)} + \bm{I}_{Tq})^{-1}.
\end{equation}
On the diagonal elements of $\bm{S}$, we impose inverse Gamma priors $s_{i}^2\sim\mathcal{G}^{-1}(a_0,b_0)$ and set $a_0 = 5$ and $b_0 = 0.4$. Since the series are standardized, this prior centers the idiosyncratic share of each series' unit variance at $10$ percent a priori, mildly favoring the common component. Draws for the variances of the idiosyncratic component can be obtained equation-by-equation from textbook inverse Gamma posteriors,
\begin{equation}
    s_i^2|\bullet \sim \mathcal{G}^{-1}\!(a_0 + T/2,\; b_0 + \sum_{t=1}^{T}(\tilde{y}_{it} - \tilde{\bm{x}}_t'\bm{\beta}_i)^2 / 2).\label{eq:idiopost}
\end{equation}

This prior setup allows for Gibbs updates of all parameters, cycling iteratively through the distributions given by (\ref{eq:VARparas}) to (\ref{eq:idiopost}). Section \ref{sec:application} of the main text reports the chain length and thinning. The factor draws are centered in sample, $\sum_{t=1}^{T}\bm{\varepsilon}_t = \bm{0}$, which is the Gaussian conditional of the factor block and fixes the location of the shocks relative to the intercept. The assumption of cross-sectionally independent idiosyncratic components is close to satisfied: the posterior mean of the correlation between the standardized idiosyncratic components averages $0.04$ in absolute value across all pairs of series, and is somewhat larger only for pairs of closely related series.

\inlinehead{External proxies and missing values} The three external shock proxies enter the VAR as additional observables, with the loadings on all shocks other than their own fixed at zero. The lag coefficients of the equations for the proxies of \citet{kanzig2021macroeconomic} and \citet{jarocinski2020deconstructing} receive a tight prior, $\mathcal{N}(0, 10^{-8})$, instead of the horseshoe, so that these proxies are essentially unpredictable from past data; the equation for the narrative shocks of \citet{romer2004new} and \citet{wieland2020financial} retains horseshoe-shrunk lag coefficients, which allows for their known predictability from past macroeconomic data. Values of the proxies outside their coverage, including the pre-sample values that enter as initial conditions and the lags at the forecast origin, are treated as missing. Conditional on the parameters and the factors, the stacked vector of missing values is jointly Gaussian, and we draw it in each sweep of the Gibbs sampler from this full conditional distribution, using the sparse precision matrix of the stacked system as in \citet{chan2023high}. The remaining steps of the sampler then condition on the completed data, and the forecasts start from the draw-specific completed initial conditions. 

Note that in our empirical application, we follow \citet{chan2023conditional} and sample from the conditional forecast distributions after running the main sampling algorithm, and thus do not condition any model parameters on the restrictions. The sampling step for the conditional forecasts can straightforwardly be moved into the main estimation algorithm to update the parameters conditional on the restrictions via data augmentation.

\subsection{Restrictions as noisy measurements}\label{app:noisy-restrictions}
An alternative to the revision in Section \ref{sec:sampling} interprets the restrictions as noisy measurements of $\bm{R}\bm{\xi}$ and updates the unconditional distribution $\bm{\xi}\sim\mathcal{N}(\bm{0},\bm{I}_d)$ via Bayes' rule, thereby combining rather than replacing the model-based forecast. When treating the restrictions as noisy observations, the posterior variance $\tilde{\bm{V}}_\xi = (\bm{I}_{d} + \bm{R}'\bm{\Omega}^{-1}\bm{R})^{-1}$  can be rewritten by applying the Woodbury matrix identity, which yields: $\tilde{\bm{V}}_\xi = \bm{I}_d - \bm{R}'(\bm{\Omega} + \bm{R}\bm{R}')^{-1}\bm{R}$. Substituting the Woodbury form of $\tilde{\bm{V}}_\xi$ into the posterior mean $\tilde{\bm{\mu}}_\xi = \tilde{\bm{V}}_\xi\bm{R}'\bm{\Omega}^{-1}\bm{r}$ yields $\tilde{\bm{\mu}}_\xi = \bm{R}'\left(\bm{\Omega}^{-1} - (\bm{\Omega}+\bm{R}\bm{R}')^{-1}\bm{R}\bm{R}'\bm{\Omega}^{-1}\right)\bm{r} = \bm{R}'(\bm{\Omega} + \bm{R}\bm{R}')^{-1}\bm{r}$. For $\bm{\Omega} = \bm{0}$, i.e., under hard restrictions, these are the same moments as in \citet{waggoner1999conditional} and identical to those in (\ref{eq:restrmoments}).

\subsection{The Kullback--Leibler divergence}\label{app:kl}
The divergence in (\ref{eq:kl}) measures how far the restricted innovations $\bm{\xi}^{\ast}$ move away from their unconditional distribution, and hence, relative to historical dynamics, how unusual the innovations are that a scenario requires. Let $p^{\ast} = \mathcal{N}(\bm{\mu}_\xi,\bm{V}_\xi)$ in (\ref{eq:restrmoments}) and $p_0 = \mathcal{N}(\bm{0},\bm{I}_d)$. By Section \ref{sec:sampling}, the revision replaces the distribution of $\bm{w} = \bm{R}\bm{\xi}$, $\mathcal{N}(\bm{0},\bm{R}\bm{R}')$ under $p_0$, by $\mathcal{N}(\bm{r},\bm{\Omega})$ under $p^{\ast}$, and leaves the conditional distribution of $\bm{\xi}$ given $\bm{w}$ unchanged. By the chain rule, the divergence is the divergence between the two distributions of $\bm{w}$ plus the expected divergence between the conditional distributions given $\bm{w}$, and the latter is zero, so
\begin{equation*}
    \mathrm{KL}(p^{\ast}\,\|\,p_0) = \mathrm{KL}\big(\mathcal{N}(\bm{r},\bm{\Omega})\,\|\,\mathcal{N}(\bm{0},\bm{R}\bm{R}')\big).
\end{equation*}
The divergence between two $r$-dimensional Gaussians,
\begin{equation*}
\begin{aligned}
    \mathrm{KL}\big(\mathcal{N}(\bm{\mu}_1,\bm{\Sigma}_1)\,\|\,\mathcal{N}(\bm{\mu}_0,\bm{\Sigma}_0)\big) = &\tfrac{1}{2}\,(\bm{\mu}_1 - \bm{\mu}_0)'\bm{\Sigma}_0^{-1}(\bm{\mu}_1 - \bm{\mu}_0)~+ \\
    \quad &\tfrac{1}{2}\big[\mathrm{tr}(\bm{\Sigma}_0^{-1}\bm{\Sigma}_1) - r - \ln\det(\bm{\Sigma}_0^{-1}\bm{\Sigma}_1)\big],
\end{aligned}
\end{equation*}
then yields (\ref{eq:kl}). Equation (16) of \citet{antolin2021structural}, applied to $\mathcal{N}(\bm{\mu}_\xi,\bm{V}_\xi)$ in dimension $d$, gives the same value but requires the $d \times d$ matrix $\bm{V}_\xi$.

The two parts have a direct interpretation. The mean part is half the squared Mahalanobis distance of the target $\bm{r}$ from the unconditional mean $\bm{0}$ of $\bm{w}$ under its unconditional variance $\bm{R}\bm{R}'$; it is large when the scenario stipulates paths the model regards as unlikely. The variance part is the divergence of $\mathcal{N}(\bm{0},\bm{\Omega})$ from $\mathcal{N}(\bm{0},\bm{R}\bm{R}')$: it is non-negative, zero if and only if $\bm{\Omega} = \bm{R}\bm{R}'$, and grows whether the scenario is stated with more or with less certainty than the model implies. Under the RF scheme with the soft constraints of Section \ref{sec:application}, $\bm{R} = \bm{R}^{(y)}\bm{M}'$ and $\bm{\Omega} = \bm{R}^{(y)}\bm{M}'\bm{M}\bm{R}^{(y)\prime} = \bm{R}\bm{R}'$, so only the mean part remains. Under the ST and DS schemes, $\bm{\Omega}$ is block-diagonal whereas $\bm{R}\bm{R}'$ is not, and the variance part is positive. Under hard constraints ($\bm{\Omega} = \bm{0}$), $p^{\ast}$ is singular with respect to $p_0$ and the divergence is infinite. A mean shift of $\delta$ standard deviations in a single innovation ($r = 1$, $\bm{R}\bm{R}' = \bm{\Omega} = 1$) has $\mathrm{KL} = \delta^2/2$, so $\sqrt{2\,\mathrm{KL}}$ is the size of a single shock with the same divergence.

Since the level of the divergence is hard to interpret, \citet{antolin2021structural} adapt \citet{mcculloch1989local} and calibrate it to coin flips: $\kappa$ solves $\mathrm{KL}\big(\mathcal{B}(m,0.5)\,\|\,\mathcal{B}(m,\kappa)\big) = \mathrm{KL}$ with $\mathcal{B}$ the binomial distribution. Since the left-hand side equals $-\tfrac{m}{2}\ln[4\kappa(1-\kappa)]$, this gives $\kappa = \tfrac{1}{2}\big(1 + \sqrt{1 - e^{-2\mathrm{KL}/m}}\big)$. Following \citet{antolin2021structural}, $m$ is the dimension of the innovation vector, $m = d$.

\subsection{Details of the computational comparison}\label{app:comptime-details}
\inlinehead{Implementation details} The labels in Figure \ref{fig:comp-time} of the main text refer to the following samplers. ``Joint (shocks)'' and ``Joint (obs.)'' refer to a direct implementation of the formulas in \citet{antolin2021structural}. These compute the Gaussian moments explicitly and factorize a dense $nH \times nH$ covariance matrix for every draw, an $\mathcal{O}((nH)^3)$ operation, in the space of structural shocks (their equations 10--12) and observables (equations 13--14), respectively. ``SVD'' is the algorithm of \citet{jarocinski2010conditional}, which accelerates \citet{waggoner1999conditional} by taking a singular value decomposition (SVD) of the $r \times nH$ restriction map, obtaining the conditional mean through the pseudo-inverse in factored form and the stochastic component from a draw in the $(nH - r)$-dimensional null space; its per-draw cost is of order $\mathcal{O}(r(nH)^2)$. ``Precision'' is the sampler of \citet{chan2023conditional}, which partitions the stacked observable vector into restricted and free coordinates and samples the free block from its Gaussian conditional using the sparse precision matrix of the stacked system; each draw requires one sparse factorization of the free-block precision of dimension $nH - r$.

We took care to make this comparison as fair as possible. That is, no method unnecessarily forms an explicit matrix inverse or the full inverse lag polynomial; whenever several algebraically equivalent construction routes exist (recursive, companion-form, and joint stacked representations), each method takes the fastest available route. For ``Precision,'' since the sparsity patterns of the precision matrix and of its free-block submatrix are fixed across draws, the sparse Cholesky factorization is initialized only once and the free block is assembled by a single vectorized copy per draw; each draw then merely updates the numeric values of the factor, which serves both the conditional mean and the simulation step. Analogously, since the sparsity structure of the restriction map of Algorithm~\ref{alg:xi_post} is fixed across draws, its assembly indices are computed once per chain, and each draw fills the map with a single vectorized copy. This bookkeeping matters because language-level overhead --- e.g., method dispatch and object validation for the sparse-matrix classes in \texttt{R} --- can dominate the actual floating-point work in small systems; we profiled all implementations and removed such overhead wherever the algorithm permits it, so that measured differences reflect each sampler's linear algebra rather than implementation artifacts.

\section{Empirical Appendix}\label{app:oa-empirical}
\subsection{Impulse responses}\label{app:irf}
Impulse response functions of selected variables, computed as described in Section \ref{sec:econometrics}, are shown in Figure \ref{fig:irf-appendix}. The figure shows the responses of selected variables to all structural shocks. The responses are computed as the median of the posterior draws, with 68\% and 90\% credible sets. They are qualitatively similar to those of the closely related work of \citet{chan2025largesign} and \citet{arias2026large}.
\begin{figure}[htbp]
  \centering
  \includegraphics[width=\textwidth]{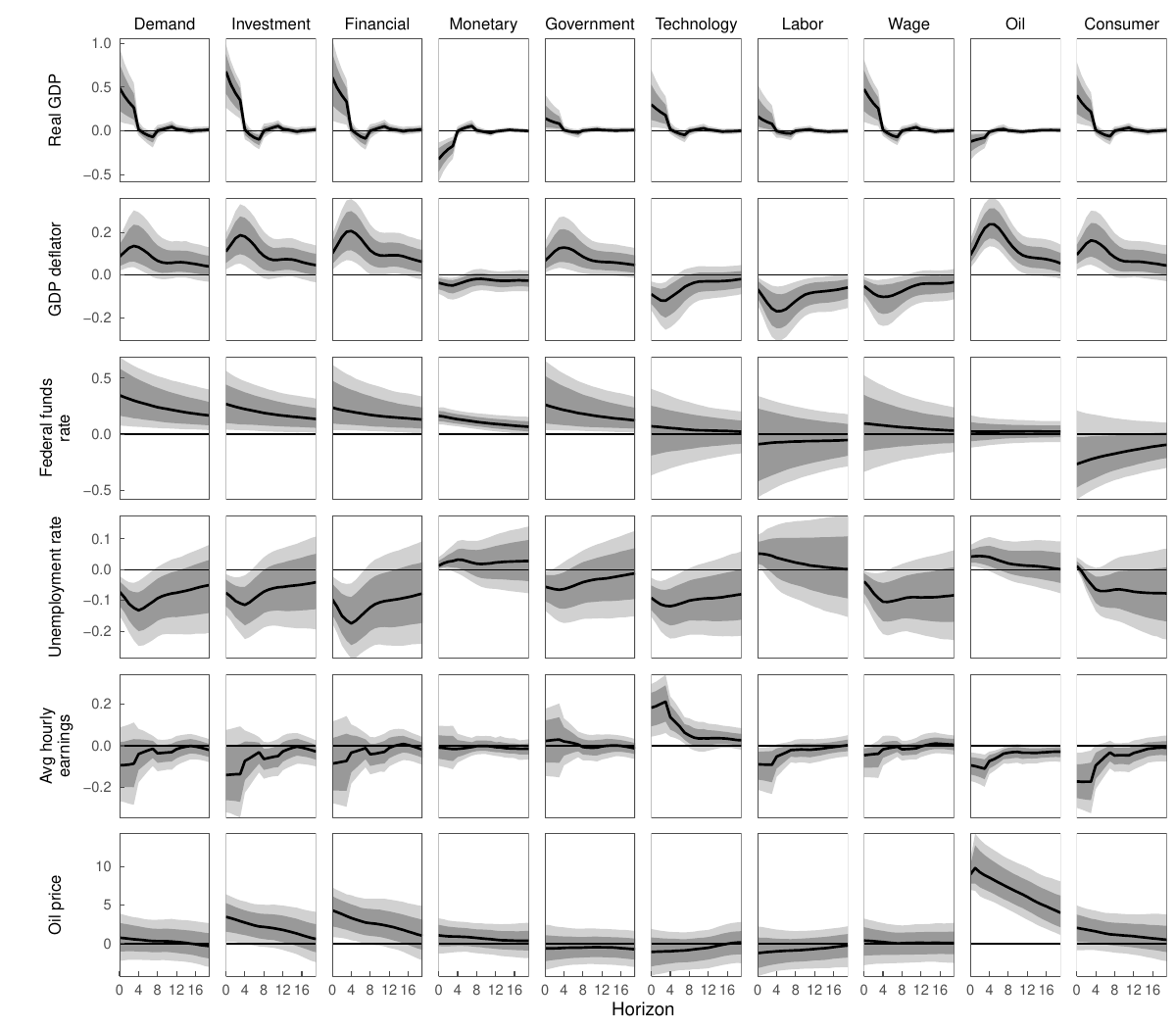}
  \caption{Impulse responses of selected variables to all structural shocks.}
  \label{fig:irf-appendix}
  \fnote{Posterior medians with 68\% and 90\% credible sets; horizon in quarters. Responses of the unemployment rate are cumulated to levels.}
\end{figure}

\subsection{Historical decomposition}\label{app:hd}
Figure \ref{fig:hd-appendix} shows the historical decomposition of real GDP growth, GDP deflator inflation and the federal funds rate. Each bar stacks the posterior medians of the contributions of the structural shocks, grouped into the oil price shock, demand-side shocks (demand, investment, financial, consumer sentiment), policy shocks (monetary policy, government spending) and supply-side shocks (technology, labor supply, wage bargaining), together with the combined contribution of the initial conditions and the idiosyncratic components.
\begin{figure}[htbp]
  \centering
  \includegraphics[width=\textwidth]{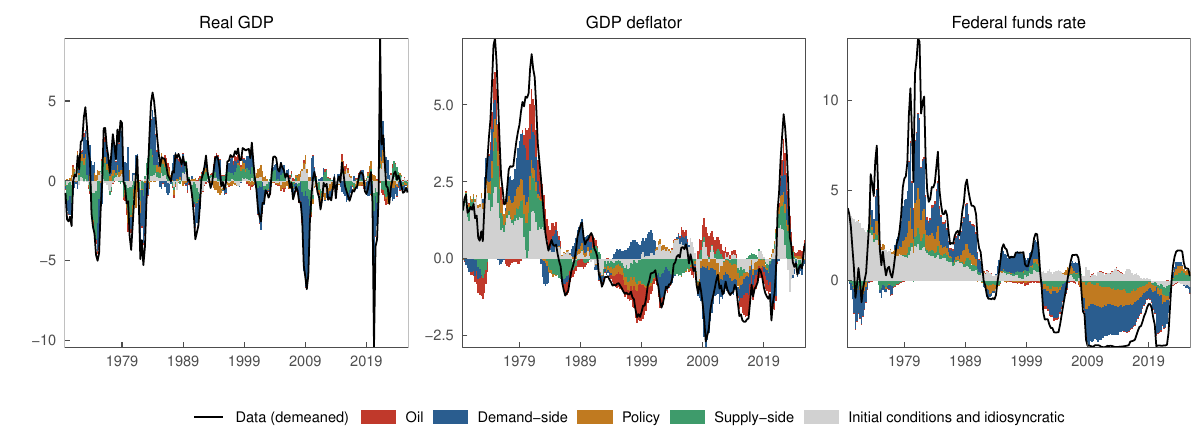}
  \caption{Historical decomposition of real GDP growth, GDP deflator inflation and the federal funds rate.}
  \label{fig:hd-appendix}
  \fnote{Bars: posterior medians of the contributions by group of shocks; black line: data net of the intercept component. Year-over-year growth and inflation and the federal funds rate in percent.}
\end{figure}

\subsection{The open idiosyncratic margin}\label{app:openmargin}
The open-margin scheme (OM) restricts all structural shocks but the oil price shock to their unconditional distribution and leaves the idiosyncratic components free. Figure \ref{fig:eps-scenarios-all} extends Figure \ref{fig:eps-scenarios} to all scenarios and restricted horizons for the structural shocks other than the oil price shock under the RF and ST schemes; Figure \ref{fig:open-margin-horizons} shows the oil price shock and the idiosyncratic component of the oil price equation under all four schemes. Under the OM scheme, the idiosyncratic component of the oil price equation absorbs more of the adjustment on impact than the oil price shock, while all other idiosyncratic components essentially remain at their unconditional distribution. The consequences for the unrestricted observables are accordingly small: at the peak, median inflation exceeds the unconditional forecast by $0.3$ to $0.8$ percentage points, and median GDP growth deviates from it by at most $0.5$ percentage points.

The scheme also has a lower divergence ($9.1$ to $15.6$) than the ST scheme ($13.1$ to $22.7$). This is the sense in which the divergence has to be read with the status of the restricted innovations in mind. It measures the tilt of the full innovation vector, and the idiosyncratic component of the oil price equation is a cheap margin along which to absorb the scenario. The resulting counterfactual attributes the oil price path to series-specific disturbances without a structural interpretation, which speaks for closing the idiosyncratic margin whenever a scenario is to be read structurally.

\begin{figure}[ht]
  \centering
  \includegraphics[width=\textwidth,height=0.8\textheight,keepaspectratio]{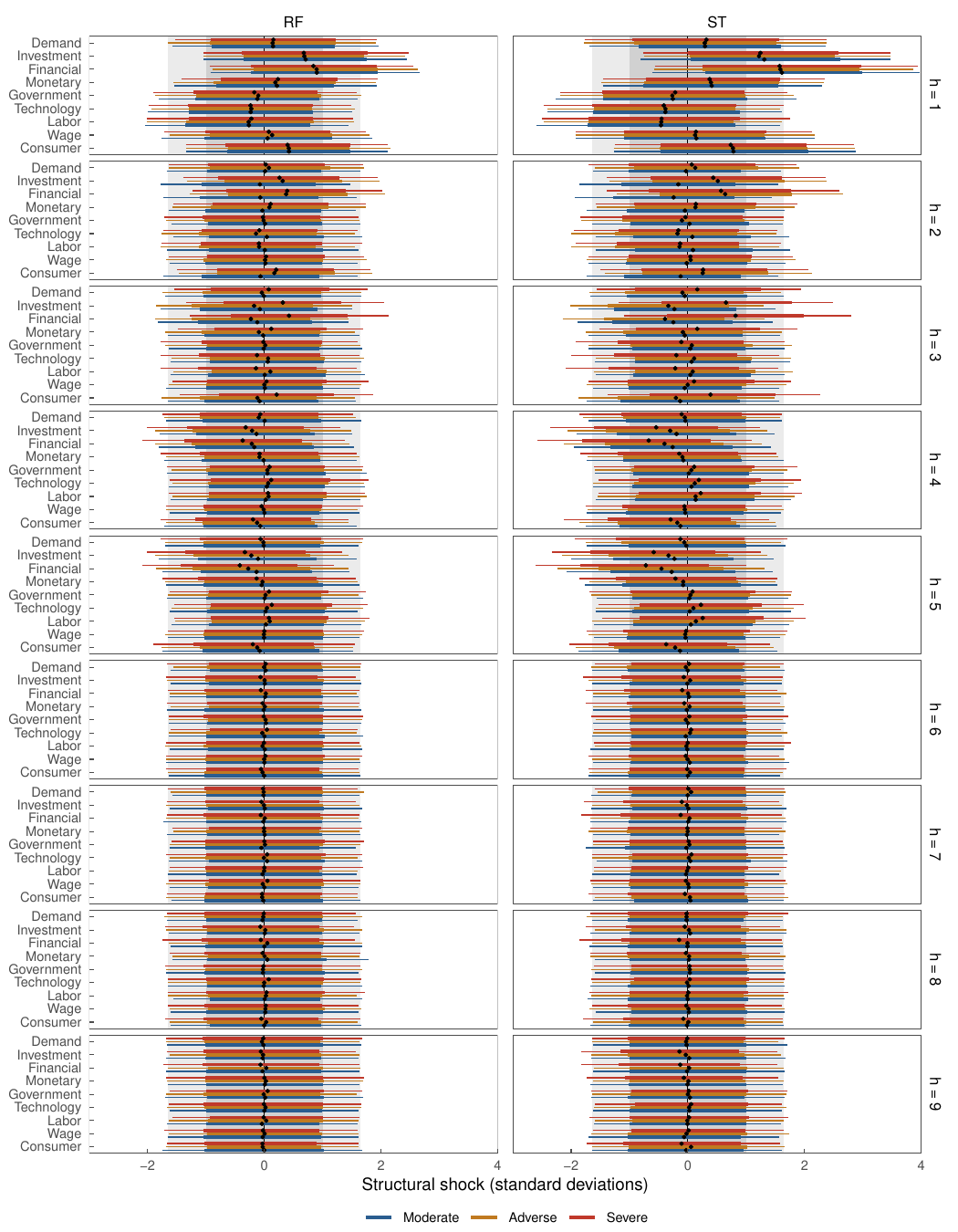}
  \caption{Required structural shocks other than the oil price shock by horizon (rows), conditioning scheme (columns) and scenario.}
  \label{fig:eps-scenarios-all}
  \fnote{RF: reduced-form, ST: structural scheme. Medians with $68$/$90$ percent intervals; grey: central $68$/$90$ percent of $\mathcal{N}(0,1)$. Under the DS and OM schemes, these shocks are restricted to $\mathcal{N}(0,1)$ by construction and therefore not shown. The oil price shock is shown in Figure \ref{fig:open-margin-horizons}.}
\end{figure}

\begin{figure}[ht]
  \centering
  \includegraphics[width=\textwidth]{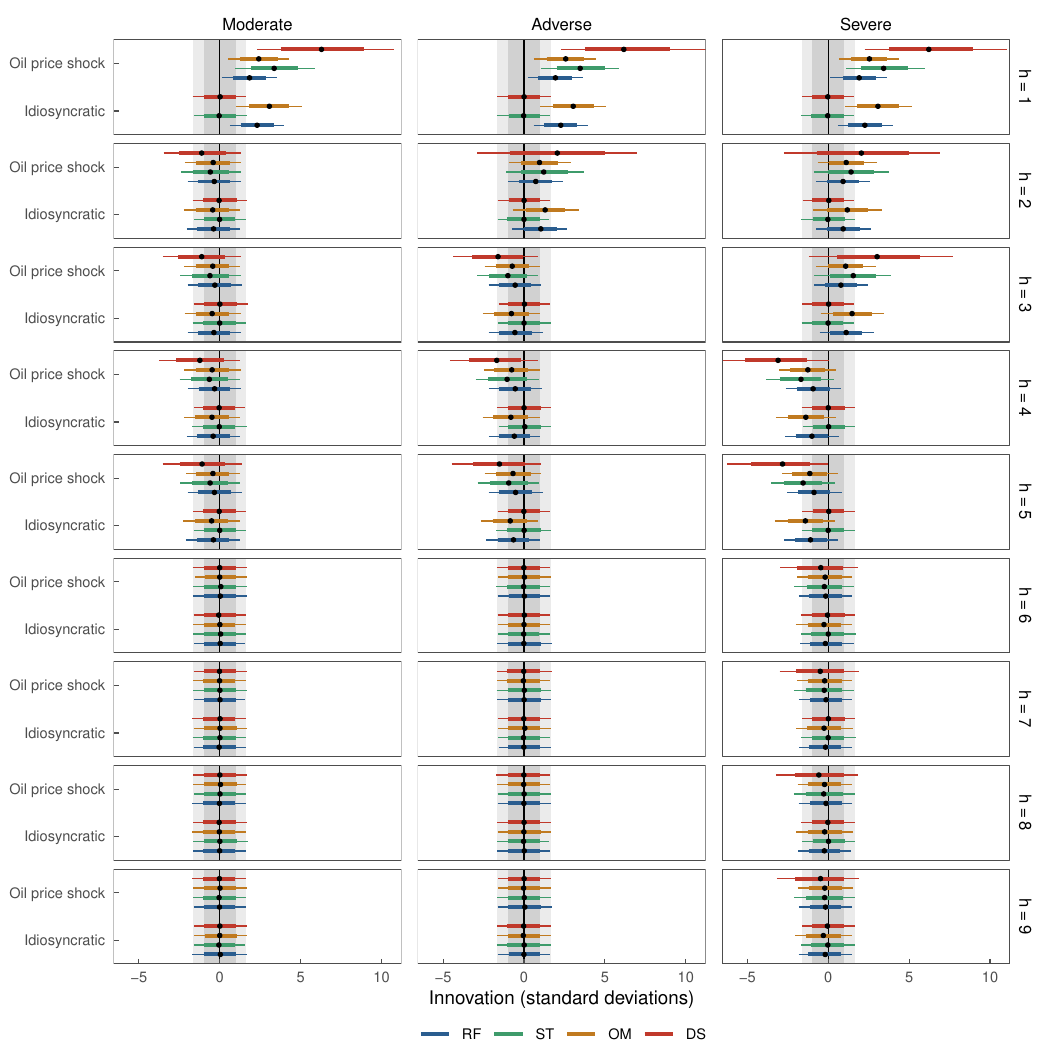}
  \caption{Structural oil price shock and idiosyncratic component of the oil price equation by horizon (rows), scenario (columns) and conditioning scheme.}
  \label{fig:open-margin-horizons}
  \fnote{RF: reduced-form, ST: structural, OM: open-margin, DS: driving-shock scheme. Medians with $68$/$90$ percent intervals; grey: central $68$/$90$ percent of $\mathcal{N}(0,1)$.}
\end{figure}

\subsection{The federal funds rate}\label{app:ffr}
Figure \ref{fig:cf-ffr} shows the conditional forecasts of the federal funds rate and its generalized impulse responses under the three conditioning schemes of the main text. Under the DS scheme, the median federal funds rate peaks $15$ to $26$ basis points above its unconditional forecast, in the quarter of the oil price peak. The systematic response is considerably stronger under the other schemes, with peaks of $33$ to $59$ basis points under the RF scheme and $61$ to $104$ basis points under the ST scheme, where the oil price path is delivered in part by other innovations, notably expansionary demand-side shocks, to which the policy rate responds more strongly, in line with what economic theory would suggest.

\begin{figure}[ht]
  \centering
  \includegraphics[width=\textwidth]{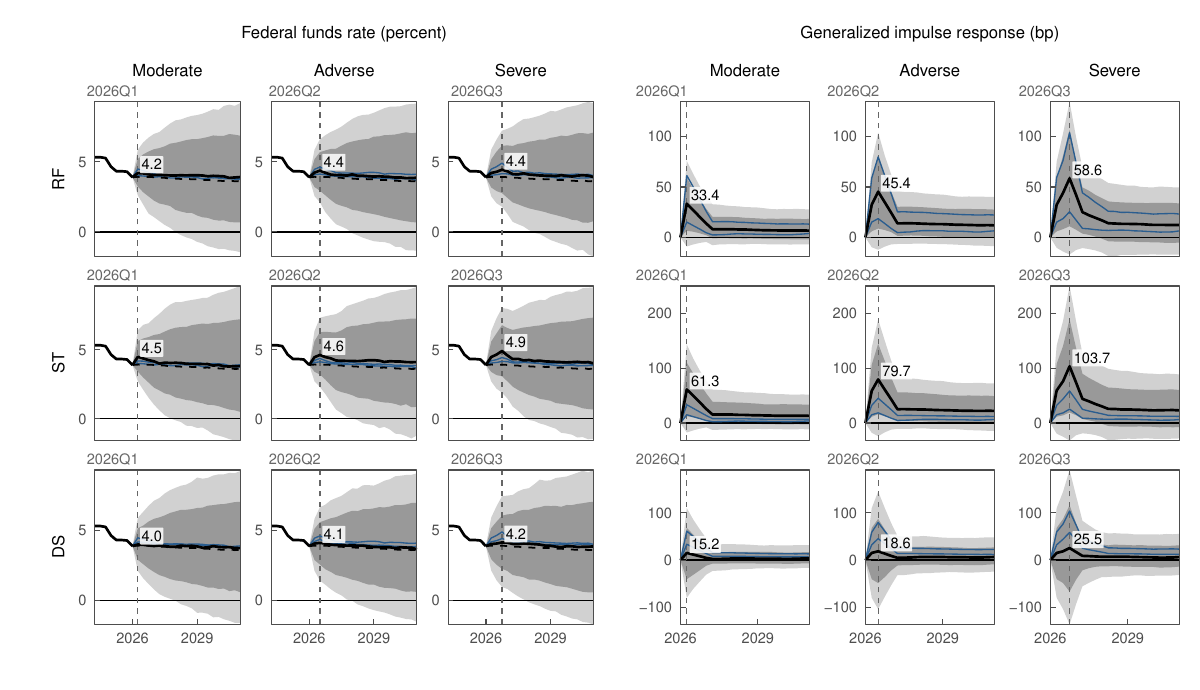}
  \caption{Conditional forecasts of the federal funds rate (left, in percent) and generalized impulse responses (right, in basis points) by scheme and scenario.}
  \label{fig:cf-ffr}
  \fnote{RF: reduced-form, ST: structural, DS: driving-shock scheme. Lines, areas and vertical markers as in Figures \ref{fig:cf-rest-types} and \ref{fig:cf-diff} of the main text; labels: median at the peak quarter. The y-axis scales differ across rows.}
\end{figure}

{\setstretch{1.2}\putbib}
\end{bibunit}

\end{document}